\documentclass[a4paper,UKenglish,cleveref, autoref, thm-restate]{lipics-v2021}

\usepackage[textsize=small, disable]{todonotes}

\usepackage{amsmath}
\usepackage{amsthm}

\usepackage[ruled,noline,linesnumbered,noend]{algorithm2e}
\usepackage{stmaryrd}
\usepackage{cleveref}

\usepackage{tikz,tikz-qtree}
\usetikzlibrary{backgrounds,calc,positioning}

\colorlet{darkgreen}{green!80!black}

\hideLIPIcs

\newcommand{\dom}{\mathop{\mathrm{dom}}}
\newcommand{\ran}{\mathop{\mathrm{ran}}}
\newcommand{\initial}{\iota}
\newcommand{\RS}{{\mathcal U}}
\DeclareMathOperator{\dist}{\mathsf{dist}}
\newcommand{\sphereStr}[3]{\mathcal{S}_{#1, #2}(#3)}
\newcommand{\neighborhoodStr}[3]{\mathcal{N}_{#1, #2}(#3)}
\DeclareMathOperator{\bigo}{\mathcal{O}}
\newcommand{\neighborhood}[2]{\mathcal{N}_{#1}(#2)}
\newcommand{\rank}{\mathsf{rank}}
\newcommand{\eval}{\mathsf{val}}
\newcommand{\DAG}{\mathsf{dag}}
\DeclareMathOperator{\numberPaths}{\#\mathsf{paths}}
\DeclareMathOperator{\edgeWeight}{\mathsf{weight}}

\newcommand{\parent}{\wp}
\newcommand{\nextlink}{\mathsf{nl}}
\newcommand{\bigR}{\rho}
\newcommand{\expan}{\mathcal{E}}
\newcommand{\lex}{\mathsf{lex}}

\title{FO-Query Enumeration over SLP-Compressed Structures of Bounded Degree}

\author{Markus Lohrey}{University of Siegen, H\"olderlinstr.~3, D-57076 Siegen, Germany}{lohrey@eti.uni-siegen.de}{https://orcid.org/0000-0002-4680-7198}{}
\author{Sebastian Maneth}{University of Bremen, Germany}{maneth@uni-bremen.de}{https://orcid.org/0000-0001-8667-5436}{}
\author{Markus L. Schmid}{Humboldt-Universit\"at zu Berlin, Unter den Linden 6, D-10099 Berlin, Germany}{MLSchmid@MLSchmid.de}{https://orcid.org/0000-0001-5137-1504}{Supported by the German Research Foundation (Deutsche Forschungsgemeinschaft, DFG) – project number 522576760 (gefördert durch die Deutsche Forschungsgemeinschaft (DFG) – Projektnummer 522576760).}

\authorrunning{M.~Lohrey, S.~Maneth, M.~Schmid}

\Copyright{Markus Lohrey, Sebastian Maneth and Markus L.~Schmid}



\category{} 

\relatedversion{} 

\nolinenumbers 

\EventEditors{Pawe\l{} Gawrychowski, Filip Mazowiecki, and Micha\l{} Skrzypczak}
\EventNoEds{3}
\EventLongTitle{50th International Symposium on Mathematical Foundations of Computer Science (MFCS 2025)}
\EventShortTitle{MFCS 2025}
\EventAcronym{MFCS}
\EventYear{2025}
\EventDate{August 25--29, 2025}
\EventLocation{Warsaw, Poland}
\EventLogo{}
\SeriesVolume{345}
\ArticleNo{39}

\begin{document}

\maketitle

\begin{abstract}
Enumerating the result set  of a first-order query over a relational structure of bounded degree can be done with linear preprocessing and constant delay. In this work, we extend this result towards the compressed perspective where the structure is given in a potentially highly compressed form by a straight-line program (SLP). Our main result is an algorithm that enumerates the result set of a first-order query over a structure
of bounded degree that is represented by an SLP satisfying the so-called apex condition.
For a fixed formula, the enumeration algorithm has constant delay and needs a preprocessing time that is linear in the size of the SLP.
\end{abstract}

\section{Introduction}

\emph{First order model checking} (i.e., deciding whether an FO-sentence $\phi$ holds in a relational structure $\RS$, $\RS \models \phi$ for short) is a classical problem in computer science and its complexity has been thoroughly investigated; see, e.g.,~\cite{FlumGrohe2006,KreutzerDawar2009,Libkin2004}. In database theory, it is of importance due to its practical relevance for evaluating SQL-like query languages in relational databases. FO model checking is {\sf PSPACE}-complete when $\phi$ and $\RS$ are both part of the input, but it becomes fixed-parameter tractable (even \emph{linear} fixed-parameter tractable) with respect to the parameter $|\phi|$ when $\RS$ is restricted to a suitable class of relational structures (see the paragraph on related work below for details),
while for the class of all structures it is not fixed-parameter tractable modulo certain complexity assumptions. This is relevant, since in practical scenarios queries are often small, especially in comparison to the data (represented by the relational structure) that is often huge. 

FO model checking (i.e., checking a Boolean query that returns either \emph{true} or \emph{false}) reduces to practical query evaluation tasks and is therefore suitable to transfer lower bounds.
However, from a practical point of view, \emph{FO-query enumeration} is more relevant. It is the problem of enumerating without repetitions for an 
FO-formula $\phi(x_1, \ldots, x_k)$ with free variables $x_1, \ldots, x_k$ the \emph{result set} $\phi(\RS)$ of all tuples $(a_1, \ldots, a_k) \in \RS^k$ such that $\RS \models \phi(a_1, \ldots, a_k)$. 
Since $\phi(\RS)$ can be rather large (exponential in $k$ in general), the total time for enumeration is not a good measure for the performance
of an enumeration algorithm. More realistic measures are the \emph{preprocessing time} (used for performing some preprocessing on the input) and the \emph{delay}, which is the maximum time needed between the production of two consecutive output tuples from $\phi(\RS)$. In \emph{data complexity} (where we consider $|\phi|$ to be constant), the best we can hope for is linear preprocessing time (i.e., $f(|\phi|) \cdot |\RS|$ for a computable function $f$) and constant delay (i.e., the delay is $f(|\phi|)$ for some computable function $f$ and therefore does not depend on $|\RS|$). Over the last two decades, many of the linear time (with respect to data complexity) FO model checking algorithms for various subclasses
of structures have been extended to  FO-query enumeration algorithms with 
 linear (or quasi-linear) time preprocessing and constant delay
(see the paragraph on related work below for the relevant literature).

In this work, we extend FO-query enumeration towards the compressed perspective, i.e., we wish to enumerate the result set $\phi(\RS)$ in the scenario where $\RS$ is given in a potentially highly compressed form, and we want to work directly on this compressed form without decompressing it. In this regard, we contribute to a recent research effort in database theory that is concerned with \emph{query evaluation over compressed data} \cite{LohreySchmid2024, MunozRiveros2023,SchmidSchweikardt2021, SchmidSchweikardt2022}. Let us now explain this framework in more detail.

\subparagraph{Query evaluation over compressed data.} Query evaluation over compressed data combines the classical task of query evaluation with the paradigm of \emph{algorithmics on compressed data} (ACD), i.e., solving computational tasks directly on compressed data objects without prior decompression. ACD is an established algorithmic paradigm and it works very well in the context of \emph{grammar-based compression} with so-called \emph{straight-line programs} (SLPs). Such SLPs use grammar-like formalisms in order to specify how a data object is constructed  from small building blocks.
For example, if the  data object  is a finite string $w$, then an SLP is just a context-free grammar for the language $\{w\}$. For instance, the SLP $S \to AA$, $A \to B B C$, $B \to ba$, $C \to cb$ (where $S, A, B,C$ are nonterminals and $a,b,c$ are terminals) produces the string $babacbbabacb$. While SLPs achieve exponential compression in the best case, there are also fast heuristic compressors that yield decent compression in practical scenarios. Moreover, SLPs are very well suited for ACD; see, e.g., \cite{Loh12survey}.

An important point is that the ACD perspective can lead to dramatic running time improvements: if the same problem can be solved in linear time both in the uncompressed and in the compressed setting (i.e., linear in the compressed size), then in the case that the input can be compressed from size $n$ to size $\bigo(\log n)$ (which is possible with SLPs in the best case), the algorithm for the compressed 
data has a running time of $\bigo(\log n)$ (compared to $\bigo(n)$ for the algorithm working on uncompressed data).
An important problem that shows this  behavior is for instance pattern matching in compressed texts~\cite{GanardiGawrychowski2022}. 

SLPs are most famous for strings (see~\cite{BannaiEtAl2021,CaselEtAl2021,GanardiGawrychowski2022,GanardiJL21} for some recent publications and~\cite{Loh12survey} for a survey).
What makes them particularly appealing for query evaluation is that their general approach of compressing data objects by grammars extends from strings to more complex structures like trees~\cite{GasconLMRS20,Lohrey15,LohreyMR18,LohreyEtAl2012} and hypergraphs (i.e., general relational structures)~\cite{LengauerWanke1988,Lohrey2012,ManethPeternek2018,ManethPeternek2020}, while, at the same time, their good ACD-properties are maintained to some extend. This is due to the fact that context-free string grammars extend to 
context-free tree grammars~\cite{DBLP:journals/mst/Rounds70} (see also~\cite{DBLP:reference/hfl/GecsegS97})
and to hyperedge replacement grammars~\cite{DBLP:journals/mst/BauderonC87,DBLP:conf/stacs/HabelK87} (see also~\cite{Eng97}). 

In this work, we are concerned with FO-query enumeration for relational structures that are compressed by SLPs based on hyperedge replacement grammars (also known as hierarchical graph definitions or SL HR grammars; see the paragraph on related work for references).
An example of such an SLP is shown in Figure~\ref{fig:SLP}. It consists of productions (shown in Figure~\ref{fig:SLP} on the left) that replace nonterminals ($S$, $A$, and $B$ in
Figure~\ref{fig:SLP}) by their unique right-hand sides.
Each right-hand side is a relational structure (a directed graph in Figure~\ref{fig:SLP}) together with occurrences of earlier defined
nonterminals and certain distinguished contact nodes (labelled by $1$ and $2$ in Figure~\ref{fig:SLP}). 
 In this way, every nonterminal $X \in \{S,A,B\}$ produces a relational structure $\eval(X)$ (the value of $X$) with distinguished contact nodes.
 These structures are shown in Figure~\ref{fig:SLP} on the right.
 When replacing for instance the occurrence of $B$ in the right hand side of $S$ by $\eval(B)$, one identifies for every $i \in \{1,2\}$
 the $i$-labelled contact node in $\eval(B)$ with the node that is connected by the $i$-labelled dotted edge with the 
 $B$-occurrence in the right-hand side of $S$ (these are the nodes labelled with $u$ and $v$ in Figure~\ref{fig:SLP}).

\subparagraph{Main result.} It is known that FO-query enumeration for degree-bounded structures can be done with linear preprocessing and constant delay~\cite{DurandGrandjean2007,KazSeg11}. Moreover, FO model checking for \emph{SLP-compressed} degree-bounded structures can be done efficiently~\cite{Lohrey2012}. We combine these two results and therefore extend FO-query enumeration for bounded-degree structures towards the SLP-compressed setting, or, in other words, we extend FO model checking of SLP-compressed structures to the query-enumeration perspective.
A preliminary version of our main result is stated below. It restricts to so-called apex SLPs. Roughly speaking, the apex property
demands that each graph replacing a nonterminal must not contain other
nonterminals at the ``contact nodes'' (the nodes the nonterminal was incident with).
 The apex property is well known from graph language theory~\cite{Eng97,EngelfrietHL94,DBLP:journals/iandc/EngelfrietR90}
 and has been used for SLPs in \cite{Lohrey2012,MHR94}. 

\begin{theorem}\label{prelimMainResultTheorem}
Let $d \in \mathbb{N}$ be a constant. For an FO-formula $\phi(x_1, \ldots, x_k)$ and 
a  relational structure $\RS$, whose Gaifman graph has degree at most $d$, and that is given in compressed form by an apex SLP $D$, 
one can enumerate the result set $\phi(\RS)$ after preprocessing time $f(d,|\phi|) \cdot |D|$ and delay $f(d,|\phi|)$ for some computable function $f$.
\end{theorem}
Note that the preprocessing is linear in the compressed size $|D|$ instead of the data size $|\RS|$.

We prove this result by extending the FO-query enumeration algorithm for uncompressed structures from 
\cite{KazSeg11} to the SLP-compressed setting. For this we have to overcome considerable technical barriers.
The algorithm of \cite{KazSeg11} exploits the Gaifman locality of FO-queries. In the preprocessing phase the 
algorithm computes for each element $a \in \RS$ the $r$-sphere around $a$ for a radius $r$ that only depends on the formula $\phi$. 
This leads to a preprocessing time of $|\RS| \cdot f(d,\phi)$. For an SLP-compressed structure we cannot afford to iterate
over all elements of the structure. Inspired by a technique from \cite{Lohrey2012}, we will expand every nonterminal of the SLP $D$ locally up to a size that depends only on $\phi$ and $d$. This leads to at most $|D|$ substructures of size $f(d,|\phi|)$. Our enumeration algorithm then 
enumerates certain paths in the derivation tree defined by $D$ and for each such path ending in a nonterminal $A$ it searches in the precomputed
local expansion of $A$ for nodes with a certain sphere type.

 \subparagraph{Related work.} In the uncompressed setting, there are several classes of relational structures for which FO-query enumeration can be solved with linear (or quasi-linear) preprocessing and constant delay, e.g., relational structures with bounded degree~\cite{BerkholzEtAl2018, DurandGrandjean2007,KazSeg11}, low degree~\cite{DurandEtAl2022}, (locally) bounded expansion~\cite{KazanaSegoufin2013, SegoufinVigny2017}, and structures that are tree-like~\cite{Bagan2006,KazanaSegoufin2013b} or nowhere dense~\cite{SchweikardtEtAl2022}; see~\cite{BerkholzEtAl2020, Segoufin2015} for surveys. Moreover, for conjunctive queries with certain acyclicity conditions, linear preprocessing and constant delay is also possible for the class of all relational structures~\cite{BaganEtAl2007, BerkholzEtAl2020}. The algorithm from~\cite{KazSeg11} is the most relevant one for our work. 

Concerning other work on query enumeration on SLP-compressed data, we mention~\cite{MunozRiveros2023,SchmidSchweikardt2021, SchmidSchweikardt2022}, which deals with constant delay enumeration for (a fragment of) MSO-queries on SLP-compressed strings, and~\cite{LohreySchmid2024}, which presents a linear preprocessing and constant delay algorithm for MSO-queries on SLP-compressed unranked forests.

SLPs for (hyper)graphs were introduced as hierarchical graph descriptions by Lengauer and Wanke~\cite{LeWa88} 
and have been further studied, e.g., in~\cite{BrenguierGS12,FaranK18,FaranK19,Len89,LeWa92,MHR94,MHSR98,MRHR97}.
Model checking problems for SLP-compressed graphs have been studied in \cite{Lohrey2012} for FO and MSO,
\cite{GoLo11} for fixpoint logics, and \cite{AlurBEGRY05,AlYa01} for the temporal logics LTL and CTL in the context of 
hierarchical state machines (which are a particular type of graph SLPs). 
Particularly relevant for this paper is a result from \cite{Lohrey2012} stating
that for every level $\Sigma^{\mathsf{P}}_i$ of the polynomial time hierarchy there is a fixed FO-formula for which the 
model checking problem for SLP-compressed input graphs is $\Sigma^{\mathsf{P}}_i$-complete. In contrast,
for apex SLPs the model checking problem for every fixed FO-formula belongs
to {\sf NL} (nondeterministic logspace) \cite{Lohrey2012}. This (and the fact that FO-query enumeration reduces to FO model checking)
partly explains the restriction to apex SLPs in Theorem~\ref{prelimMainResultTheorem}.

Compression of graphs via graph SLPs has been considered in~\cite{DBLP:conf/dcc/Peshkin07} following
a ``Sequitur scheme''~\cite{DBLP:journals/jair/Nevill-ManningW97} 
and in~\cite{ManethPeternek2020} following 
a ``Repair scheme''~\cite{DBLP:journals/pieee/LarssonM00} (see also~\cite{DBLP:journals/is/LohreyMM13});
note that both compressors produce graph SLPs that may \emph{not} be apex.

Another recent concept in database theory that is concerned with compressed representations of relational data and query evaluation are  \emph{factorized databases} (see~\cite{KimelfeldEtAl2023arxiv,Olteanu2016,OlteanuSchleich2016,OlteanuSchleich2016b,OlteanuZavodny2015}). Intuitively speaking, in a factorized representation of a relational structure each relation $R$ is represented as an expression over the relational operators union and product that evaluates to $R$. However, SLPs for relational structures and factorized representations cover completely different aspects of redundancy: A factorized representation is always at least as large as its active domain (i.e., all elements that occur in some tuple), while an SLP for a relational structure can be of logarithmic size in the size of the universe of the structure. On the other hand, small factorized representations do not seem to necessarily translate into small SLPs.

\section{General Notations}\label{sec:generalNotations}

Let $\mathbb{N} = \{0, 1, 2, \ldots\}$. For every $k \in \mathbb{N}$, we set $[k] = \{1, 2, \ldots, k\}$. For a finite alphabet $A$, we denote by $A^*$ the set of all finite strings over $A$ including the empty string $\varepsilon$.
For a partial $f : A \to B$ let $\dom(f) = \{a \in A : f(a) \neq \bot\} \subseteq A$ (where $\bot \notin B$
stands for undefined) and $\ran(f) = \{ f(a) : a \in \dom(f) \} \subseteq B$.
For functions $f : A \to B$ and $g : B \to C$ we define
the composition $g \circ f : A \to C$ by $(g \circ f)(a) = g(f(a))$ 
for all $a \in A$.

A \emph{partial $k$-tuple} over a set $A$ is a partial function $t : [k] \to A$.
If $\dom(t) = [k]$, then we also say that $t$ is a \emph{complete} $k$-tuple or just a $k$-tuple; in this case we also write $t$
in the conventional form $(t(1), t(2), \ldots, t(k))$. 
Two partial $k$-tuples $t_1$ and $t_2$ are \emph{disjoint} if $\dom(t_1) \cap \dom(t_2) = \emptyset$.
In this case, their union $t_1 \sqcup t_2$ is the partial $k$-tuple defined by $(t_1 \sqcup t_2)(j) = t_i(j)$ if
$j \in \dom(t_i)$ for $i \in \{1,2\}$ and $(t_1 \sqcup t_2)(j) = \bot$ if $j \notin \dom(t_1) \cup \dom(t_2)$.

\subsection{Directed acyclic graphs}
\label{sec-dag}

An \emph{ordered dag} (directed acyclic graph) is a triple $G = (V, \gamma, \initial)$, where $V$ is a finite set of nodes, $\gamma : V \to V^*$ is the child-function, the relation $E := \{ (u,v) : u,v\in V, v \text{ occurs in } \gamma(u) \}$ is acyclic, and $\initial \in V$ is the \emph{initial node}. 
The size of $G$ is $|G|= \sum_{v \in V}(1 + |\gamma(v)|)$. A node $v \in V$ with $|\gamma(v)|=0$ is called a \emph{leaf}. 

A path in $G$ (from $v_0$ to $v_n$) is a sequence $p = v_0 i_1 v_1 i_2 \cdots v_{n-1} i_n v_n \in V (\mathbb{N}V)^*$
such that $1 \le i_k \le |\gamma(v_{k-1})|$ for all $k \in [n]$. The length of this path $p$ is $n$ (we may have $n=0$ in which case
$p = v_0$) and we also call $p$ a \emph{$v_0$-to-$v_n$ path}
or $v_0$-path if the end point $v_n$ is not important.
An $\iota$-path is also called an \emph{initial path}. We extend this notation to subsets of $V$ in the obvious way,
e.g., for $A, B \subseteq V$ and $v \in V$ we talk about $A$-to-$v$ paths, $A$-to-$B$ paths, $A$-to-leaf paths (where ``leaf'' refers to the set of all leaves of the dag), initial-to-leaf paths, etc.  For a $v_0$-to-$v_1$ path $p = p' v_1$ and a $v_1$-to-$v_2$ path $q = v_1 q'$ we define
the $v_0$-to-$v_2$ path $pq = p' v_1 q'$ (note that if we just concatenate $p$ and $q$ as words, then we have to replace the double occurrence
$v_1 v_1$ by $v_1$ to obtain $pq$).
We say that the path $p$ is a \emph{prefix} of the path $q$ if there is a path $r$ such that $q = qr$.

Since we consider ordered dags, there is a natural lexicographical ordering on all $v$-paths (i.e., all paths that start in the same node $v$). More precisely, for two different $v$-paths $p$ and $q$ we write 
$p < q$ if either $p$ is a proper prefix of $q$ or we can write $p$ and $q$ as $p = r i p'$, $q = r j q'$ for paths $r, p', q'$ and $i, j \in \mathbb{N}$
with $i <j$.

\subsection{Relational structures and first order logic}\label{sec-FOL}

A \emph{signature} $\mathcal R$ is a finite set consisting of relation
symbols $r_i$ ($i \in I$) and constant symbols $c_j$ ($j \in J$).  
Each relation symbol $r_i$ has an associated arity $\alpha_i$. 
A \emph{structure over the signature} $\mathcal R$ is a tuple 
$\RS = (U, (R_i)_{i \in I}, (u_j)_{j \in J})$, where
$U$ is a finite non-empty set (the universe of $\RS$),
$R_i \subseteq U^{\alpha_i}$ is the relation associated with 
the relation symbol $r_i$, and $u_j \in U$ is the constant associated
with the constant symbol $c_j$. Note that we restrict to finite structures.
If the structure $\RS$ is clear from the context, we will 
identify $R_i$ ($u_j$, respectively) with the relation symbol $r_i$ 
(the constant symbol $c_j$, respectively). 
Sometimes, when we want to refer to the universe $U$, 
we will refer to $\RS$ itself. For instance, we write $a \in \mathcal U$ 
for $ua \in U$, or $f : [n] \to \mathcal U$ for a function 
$f : [n] \to U$.
The size $|\RS|$ of $\RS$ is 
$|U| + \sum_{i \in I} \alpha_i \cdot |R_i|$.
As usual, a constant $a \in \mathcal U$ may be replaced by 
the unary relation $\{a\}$.
Thus, in the following, we will only consider
signatures 
without constant symbols, except when we explicitly
introduce constants.  
Let $\mathcal R = \{ r_i : i \in I\}$
be such a signature (we call it a {\em relational signature}) and let $\RS = (U, (R_i)_{i \in I})$
be a structure over $\mathcal R$ (we call it a {\em relational structure}).
For relational structures $\RS_1$ and $\RS_2$ over the signature $\mathcal{R}$, we write $\RS_1 \simeq \RS_2$ to denote that $\RS_1$ and $\RS_2$ are isomorphic. A substructure of $\RS = (U, (R_i)_{i \in I})$ is a structure $(V, (S_i)_{i \in I})$ such that
$V \subseteq U$ and $S_i \subseteq R_i$ for all $i \in I$. The substructure of $\RS$ \emph{induced} by $V \subseteq U$ is
$(V, (R_i \cap V^{\alpha_i})_{i \in I})$.
We define the undirected graph $\mathcal{G}(\RS) = (U, E)$ (the so-called Gaifman graph of $\RS$),
where $E$ contains an edge $(a,b)$ if and only
if there is a binary relation $R_i$ ($i \in I$) and a tuple $(a_1, \ldots, a_{\alpha_i}) \in R_i$ with 
$\{a,b\} \subseteq \{a_1, \ldots, a_{\alpha_i}\}$.
The degree of $\RS$ is the maximal degree of a node in $\mathcal{G}(\RS)$. 
If $\RS$ has degree at most $d$, we also say that $\RS$ is a \emph{degree-$d$ bounded structures}.

We use \emph{first-order logic} (FO) over finite relational structures; see~\cite{EbbF91} for a detailed introduction and Appendix~\ref{sec:generalDefinitionsAppendix} for some standard notations. For an FO-formula $\psi(x_1, \ldots, x_k)$ over the signature $\mathcal R$ with free variables $x_1, \ldots, x_k$ and a relational structure $\RS = (U, (R_i)_{i \in I})$ over $\mathcal{R}$ and $a_1, \ldots, a_k \in U$, we write $\RS \models \psi(a_1, \ldots, a_k)$ if $\psi$ is true in $\RS$ when the variable $x_i$ is set to $a_i$ for all $i \in [k]$.
Hence, an FO-formula $\psi(x_1, \ldots, x_k)$  can be interpreted as an \emph{FO-query} that, for a given structure $\RS$, defines a \emph{result set} \[\psi(\RS) = \{(a_1, \ldots, a_k) \in \RS^k : \RS \models \psi(a_1, \ldots, a_k)\}.\] 
The \emph{quantifier rank} $\mathsf{qr}(\psi)$  of an FO-formula $\psi$ is the maximal nesting depth of quantifiers in $\psi$.

In the rest of the paper, we assume that the signature only contains relation symbols of arity at most two. It is folklore that FO model checking and FO-query enumeration over arbitrary signatures can be reduced to this case; see Appendix~\ref{sec:generalDefinitionsAppendix} for a possible construction. This construction can be carried out in linear time (in the size of the formula and the structure) and it increase the degree of the 
structure as well as the quantifier rank of the formula by at most one.

\subsection{Distances, spheres and neighborhoods}

Let us fix a relational signature $\mathcal{R}$ (containing only relation symbols of arity at most two)
and let $\RS = (U, (R_i)_{i \in I})$ be a structure over this signature.
We say that $\RS$ is \emph{connected}, if its Gaifman graph $\mathcal{G}(\RS)$ is connected.
The distance between elements $a,b \in U$ in the graph $\mathcal{G}(\RS)$ is denoted by $\dist_{\RS}(a, b)$
(it can be $\infty$).
For subsets $A, B \subseteq U$ we define $\dist_{\RS}(A, B) = \min \{ \dist_{\RS}(a, b) : a \in A, b \in B\}$. 
For two partial tuples (of any arity) $t, t'$ over $U$ let
$\dist_{\RS}(t, t') = \dist_{\RS}( \ran(t), \ran(t'))$.

Fix a constant $d \geq 2$. We will only consider degree-$d$ bounded structures in the following. Let us fix such a structure $\RS$ (over the 
relational signature $\mathcal{R}$). Take additional constant symbols $c_1, c_2, \ldots$ called \emph{sphere center constants}.
For an $r \geq 1$ and a partial $k$-tuple $t : [k] \to \RS$ we define the $r$-sphere
$\sphereStr{\RS}{r}{t} = \{b \in \RS : \dist_{\RS}(t, b) \leq r \}$.
The \emph{$r$-neighborhood} $\neighborhoodStr{\RS}{r}{t}$ of $t$ is obtained by taking the substructure of $\RS$  
induced by $\sphereStr{\RS}{r}{t}$ and then adding every node $t(i)$ ($i \in \dom(t)$) as the interpretation of the 
 sphere center constant $c_i$. Hence, it is a structure over the 
signature $\mathcal{R} \cup \{ c_i : i \in \dom(t)\}$. 
The $r$-neighborhood  of a $k$-tuple has at most $k \cdot \sum^{r}_{i = 0} d^{i} \le  k \cdot d^{r+1}$ elements (here, the inequality holds
since we assume $d \ge 2$).

We use the above definitions also for a single element $a \in \RS$ in place of a tuple $t$; formally
$a$ is identified with the $1$-tuple $t$ such that $t(1)=a$.
We are mainly interested in $r$-spheres and $r$-neighborhoods of complete $k$-tuples, but 
the corresponding notions for partial $k$-tuples will be convenient later.
We also drop the subscript $\RS$ from the above notations if this does not cause any confusion. 

A \emph{$(k,r)$-neighborhood type} is an isomorphism type for the $r$-neighborhood of a complete $k$-tuple in a 
 degree-$d$ bounded structure. More precisely, we can define a
 $(k,r)$-neighborhood type as a degree-$d$ bounded structure $\mathcal{B}$ over the signature 
 $\mathcal{R} \cup \{c_1, \ldots, c_k\}$ such that
 \begin{itemize}
 \item the universe of $\mathcal{B}$ is of the form $[\ell]$ for some $\ell \leq k \cdot d^{r+1}$ and
 \item for every $j \in [\ell]$ there is $i \in [k]$ such that $\dist_{\mathcal{B}}(a_i,j) \leq r$, where, for every $i \in [k]$, $a_i$ is the interpretation of the sphere center constant $c_i$.
  \end{itemize} 
  From each isomorphism class of $(k,r)$-neighborhood types we select a unique representative
  and write $\mathcal{T}_{k,r}$ for the set of all selected representatives. Then, for every $k$-tuple
   $\bar{a} \in \RS^k$ there is a unique $\mathcal{B} \in \mathcal{T}_{k,r}$ such that 
  $\neighborhoodStr{\RS}{r}{\bar{a}} \simeq \mathcal{B}$; we call it the $(k,r)$-neighborhood type of 
  $\bar{a}$ and say that $\bar{a}$ is a \emph{$\mathcal{B}$-tuple}. In case $k=1$ we speak of 
  \emph{$\mathcal{B}$-nodes} instead of $\mathcal{B}$-tuples, write $\mathcal{T}_{r}$ for $\mathcal{T}_{1,r}$
  and call its elements $r$-neighborhood types instead of 
   $(1,r)$-neighborhood types.
    
  For every $(k,r)$-neighborhood type $\mathcal{B} \in \mathcal{T}_{k,r}$ there is 
  an FO-formula $\psi^{\mathcal{B}}(x_1, \ldots, x_k)$ such that for every  
  degree-$d$ bounded structure $\RS$ and every $k$-tuple $\bar{a} \in \RS^k$ we have
  $\RS \models \psi^{\mathcal{B}}(\bar{a})$ if and only if $\bar{a}$ is a $\mathcal{B}$-tuple.

\subsection{Enumeration algorithms and the machine model}

 \emph{FO-query enumeration} is the following problem: Given an FO-formula $\phi(x_1, \ldots, x_k)$ over some signature $\mathcal{R}$ and a relational structure $\RS$ over $\mathcal{R}$, we want to enumerate all tuples from $\phi(\RS)$ in some order and without repetitions, i.e., we want to produce a sequence $(t_1, \ldots, t_n, t_{n+1})$ with $\{t_1, \ldots, t_n\} = \phi(\RS)$, $|\phi(\RS)| = n$ and $t_{n+1} = \mathsf{EOE}$ is the \emph{end-of-enumeration} marker. An algorithm for FO-query enumeration starts with a \emph{preprocessing phase} in which no output is produced, followed by an \emph{enumeration phase}, where the elements $t_1, t_2, \ldots, t_n, t_{n + 1}$ are produced one after the other. The running time of the preprocessing phase is called the \emph{preprocessing time}, and the \emph{delay} measures the maximal time between the computation of two consecutive outputs $t_i$ and $t_{i+1}$ for every $i \in [n]$. 

Usually, one restricts the input structure $\RS$ to some subclass $\mathsf{C}_d$ of relational structures that is defined by some parameter $d$ (in this paper, $\mathsf{C}_d$ is the class of degree-$d$ bounded
structures).  
We say that an algorithm for FO-query enumeration for $\mathsf{C}_d$ has \emph{linear preprocessing} and \emph{constant delay}, if its preprocessing time is $\bigo(|\RS| \cdot f(d,|\phi|))$ and its delay is $\bigo(f(d,|\phi|))$ for some computable function $f$. This complexity measure where the query $\phi$ is considered to be constant and the running time is only measured in terms of the data, i.e., the size of the relational structure, is also called \emph{data complexity}. In data complexity, linear preprocessing and constant delay is considered to be optimal (since we assume that the relational structure has to be read at least once). As mentioned in the introduction, FO-query enumeration can be solved with linear preprocessing and constant delay for several classes $\mathsf{C}_d$.

For proving upper bounds in data complexity, we often have to argue that certain computational tasks can be performed in time $f(\cdot)$ (or $|\RS| \cdot f(\cdot))$ for some function $f$. In these cases, without explicitly mentioning this in the remainder, $f$ will always be a computable function (actually, $f$ will be elementary, i.e., bounded by an exponent tower of fixed height). The arguments for $f$ will only depend on the parameter $d$ and the formula size $|\phi|$.

The special feature of this work is that we consider FO-query enumeration in the setting where the relational structure $\RS$ is not given explicitly, but in a potentially highly compressed form, and our enumeration algorithm must handle this compressed representation rather than decompressing it. Then the structure size $|\RS|$ will be replaced by the size of the compressed representation of $\RS$ in all time bounds.
 This aspect shall be explained in detail in Section~\ref{S hier}.

We use the standard RAM model with uniform cost measure as our model of computation. We will make some further restrictions for the register length tailored to the compressed setting in Section~\ref{sec:SLPRAM}.

\section{FO-Enumeration over Uncompressed Degree-Bounded Structures}\label{sec:uncompressed}

In this section, we fix a relational signature $\mathcal{R} = \{R_i : i \in I\}$, constants $d \geq 2$ and $\nu$, a degree-$d$ bounded 
structure $\RS = (U, (R_i)_{i \in I})$ over the signature $\mathcal{R}$, and an FO-formula $\phi(x_1, \ldots, x_k)$ over the signature $\mathcal{R}$ with $\mathsf{qr}(\phi) = \nu$.
Our goal  is to enumerate the set $\phi(\RS)$ after a linear time preprocessing in constant delay. Before we consider the case
where the structure $\RS$ is given in a compressed form, we will first outline the enumeration algorithm from \cite{KazSeg11} for 
 the case where $\RS$ is given explicitly (with some modifications). In Section~\ref{sec-enum-compressed} we will extend this algorithm to the compressed 
 setting.

By a standard application of the Gaifman locality of FO (see Appendix~\ref{sec-gaifman}), we first reduce the enumeration 
of $\phi(\RS)$ to the enumeration of all $\mathcal{B}$-tuples from $\RS^k$ for a fixed
 $\mathcal{B} \in \mathcal{T}_{k,r}$ (for some $r \leq 7^\nu$).
Recall that $\mathcal{B}$ contains at most $k \cdot d^{r+1}$ elements, and this upper bound only depends on $d$ and the formula $\phi$.
To simplify notation, we assume that in $\mathcal{B}$ the sphere center constant $c_i$ is interpreted by $i \in [k]$.
In particular, the sphere center constants are interpreted by different elements. This is not a real restriction; see
Appendix~\ref{appendix-pairwise-diff}.

In order to enumerate all $\mathcal{B}$-tuples, we will factorize $\mathcal{B}$ into its connected components. In order to accomplish this, we need the following definitions.
We first define the larger radius 
\begin{equation} \label{def-bigR}
\bigR = 2rk - r + k -1.
\end{equation}
Every $\bigR$-neighborhood of an element $a \in \RS$ 
has at most $d^{\bigR+1}$ elements. 
Recall that a $\bigR$-neighborhood type is a degree-$d$ bounded structure over the signature $\mathcal{R}_1 := \mathcal{R} \cup \{ c_1 \}$ with a universe $[\ell]$ for some $\ell \leq d^{\bigR+1}$. W.l.o.g. we assume that the sphere center constant $c_1$ is interpreted by the element $1$ in 
a $\bigR$-neighborhood type. Hence, every
$j \in [\ell]$ has distance at most $\bigR$ from $1$. Moreover, the $\bigR$-neighborhood types in $\mathcal{T}_{\bigR}$ are pairwise non-isomorphic.

Assume that our fixed $(k,r)$-neighborhood type $\mathcal{B}$ splits into $m \geq 1$ connected components $\mathcal{C}^{\mathcal{B}}_1,
\ldots, \mathcal{C}^{\mathcal{B}}_m$. Thus, every $\mathcal{C}^{\mathcal{B}}_i$ is a connected induced substructure of $\mathcal{B}$, every node of $\mathcal{B}$ belongs to exactly one $\mathcal{C}^{\mathcal{B}}_i$,
and there is no edge in the undirected graph $\mathcal{G}(\mathcal{B})$
between two different components $\mathcal{C}^{\mathcal{B}}_i$. Let $D_i = \mathcal{C}^{\mathcal{B}}_i \cap [k]$ be the set
of sphere centers that belong to the connected component $\mathcal{C}^{\mathcal{B}}_i$. We must have $D_i \neq \emptyset$.
Let $n_i = \min(D_i)$ (we could also fix any other 
element from $D_i$). Every node in $\mathcal{C}^{\mathcal{B}}_i$ has distance at most $r$ from some $j \in D_i$. Since 
$\mathcal{C}^{\mathcal{B}}_i$ is connected, it follows that
every node in $\mathcal{C}^{\mathcal{B}}_i$ has distance at most $r + (k-1) (2r+1) = 2rk - r + k -1 = \bigR$
from $n_i$ (this is in fact true for every $j \in D_i$ instead of $n_i$).
A \emph{consistent factorization} of our $(k,r)$-neighborhood type $\mathcal{B}$ is a 
 tuple 
 \[\Lambda = (\mathcal{B}_1, \sigma_1, \mathcal{B}_2, \sigma_2, \ldots, \mathcal{B}_m, \sigma_m)\] 
 with the following properties for all $i \in [m]$:
 \begin{itemize}
 \item $\mathcal{B}_i \in \mathcal{T}_{\bigR}$ and $\sigma_i : [k] \to \mathcal{B}_i$ is a partial $k$-tuple with $\dom(\sigma_i) = D_i$ 
and $\sigma_i(n_i) = 1$ (so, $n_i$ is mapped by $\sigma_i$ to the center of $\mathcal{B}_i$) and
\item $\neighborhoodStr{\mathcal{B}_i}{r}{\sigma_i} \simeq \mathcal{C}^{\mathcal{B}}_i$.
 \end{itemize}
 Clearly, the number of possible consistent factorizations of $\mathcal{B}$ is bounded by $f(d,|\phi|)$.
 
 For a $\bigR$-neighborhood type $\mathcal{B}'$, a $\mathcal{B}'$-node $a \in \RS$ and 
 a partial $k$-tuple $\sigma : [k] \to \mathcal{B}'$ we 
 moreover fix an isomorphism $\pi_a : \mathcal{B}' \to \neighborhoodStr{\RS}{\bigR}{a}$ (this isomorphism is not necessarily unique)
and define the partial $k$-tuple $t_{a,\sigma} : [k] \to \RS$ by
 $t_{a,\sigma}(j) = \pi_a(\sigma(j))$ for all $j \in \dom(\sigma)$.
Note that, by definition, $\pi_a(1) = a$. 
  
 Take a consistent factorization $\Lambda = (\mathcal{B}_1, \sigma_1, \ldots, \mathcal{B}_m, \sigma_m)$ 
 of $\mathcal{B}$. 
 We say that an $m$-tuple $(b_1, \ldots, b_m) \in \RS^m$ is \emph{admissible} for $\Lambda$ if the following conditions hold:
 \begin{itemize}
 \item for all $i \in [m]$, $b_i$ is a $\mathcal{B}_i$-node, and
 \item for all $i, j  \in [m]$ with $i \neq j$ we have
 \begin{equation} \label{disjointness}
\dist_{\RS}(t_{b_i,\sigma_i}, t_{b_{j},\sigma_{j}}) > 2r+1.
\end{equation}
 \end{itemize}
 Finally, with an $m$-tuple $\bar{b} = (b_1, \ldots, b_m)$ we associate the $k$-tuple 
 \[\Lambda(\bar{b}) = t_{b_1,\sigma_1} \sqcup t_{b_2,\sigma_2} \sqcup \cdots
 \sqcup t_{b_m,\sigma_m}.\] 
 Note that $t_{b_i,\sigma_i}(n_i) = \pi_{b_i}(\sigma_i(n_i)) = \pi_{b_i}(1) = b_i$.

 We claim that in order to enumerate all $\mathcal{B}$-tuples $\bar{a} \in \RS^k$, it suffices to enumerate for every consistent factorization
 $\Lambda = (\mathcal{B}_1, \sigma_1,  \ldots, \mathcal{B}_m, \sigma_m)$ of $\mathcal{B}$ the set
 of all $m$-tuples $\bar{b} \in \RS^m$ that are admissible for $\Lambda$. If we can do this,
 then we replace every output tuple $\bar{b} \in \RS^m$ by $\Lambda(\bar{b}) \in \RS^k$. Note that 
 $\Lambda(\bar{b})$ can be easily computed in time $\bigo(k)$ from the tuple $\bar{b}$, the isomorphisms $\pi_{b_i}$, and the partial
 $k$-tuples $\sigma_i : [k] \to \mathcal{B}_i$. The correctness of this algorithm follows from the following two lemmas (with full proofs in 
 Appendix~\ref{lemma-2-3}):

 \begin{lemma}\label{correctnessLemma1}
 If $\Lambda$ is a consistent factorization of $\mathcal{B}$
 and $\bar{b} \in \RS^m$ is admissible for $\Lambda$ then $\Lambda(\bar{b}) \in \RS^k$
is a $\mathcal{B}$-tuple.
 \end{lemma}
  
 \begin{lemma}\label{correctnessLemma2}
 If $\bar{a} \in \RS^k$ is a $\mathcal{B}$-tuple then there are a unique consistent factorization $\Lambda$ of 
$\mathcal{B}$ and a unique $m$-tuple $\bar{b} \in \RS^m$ that is admissible for $\Lambda$
 and such that $\bar{a}  = \Lambda(\bar{b})$.
\end{lemma}

\subsection{Enumeration algorithm for uncompressed structures}\label{sec:enumAlgo}

Let us fix a $(k,r)$-neighborhood type $\mathcal{B}$ and a consistent factorization $\Lambda = (\mathcal{B}_1, \sigma_1, \ldots, \mathcal{B}_m, \sigma_m)$ of $\mathcal{B}$.
By Lemmas~\ref{correctnessLemma1} and \ref{correctnessLemma2}, it suffices
 to enumerate (with linear preprocessing and constant delay) the set of all $\bar{a} \in \RS^m$ that are admissible for $\Lambda$.
 In the preprocessing phase we compute
\begin{itemize}
\item for every $i \in [m]$ a list $L_i$ containing all $\mathcal{B}_i$-nodes from $\RS$ and
\item for every $a \in L_i$ an isomorphism $\pi_a : \mathcal{B}_i \to \neighborhood{\bigR}{a}$.
\end{itemize}
It is straightforward to compute these data in time $|\RS| \cdot f(d, |\phi|)$
(in Section~\ref{sec-enum-compressed}, where we deal with the more general SLP-compressed case, this is more subtle).
We classify each list $L_i$ as being \emph{short} if $|L_i| \leq 
k \cdot d^{2\bigR+2r+2}$ and as being \emph{long} otherwise.
Without loss of generality, we assume that, for some $0 \leq q \leq m$ the lists $L_{1}, \ldots, L_{q}$ are short and the lists $L_{q + 1}, \ldots, L_{m}$ are long (note that this includes the cases that all lists are short or all lists are long).

Our enumeration procedure maintains a stack of the form $a_1 a_2 \cdots a_\ell$ with $0 \le \ell \leq m$ and $a_i \in L_i$ for all $i \in [\ell]$. Note that if $\ell = 0$, then we have the empty stack $\varepsilon$. Such a stack is called \emph{admissible} for $\Lambda$ (or just admissible), 
if for all $i, i'  \in [\ell]$ with $i \neq i'$ and all $j \in \dom(\sigma_i)$ and $j' \in \dom(\sigma_{i'})$ we have
$\dist_{\RS}(\pi_{a_i}(\sigma_i(j)), \pi_{a_{i'}}(\sigma_{i'}(j'))) > 2r+1$.
Note that the empty stack as well as every stack $a_1$ with $a_1 \in L_{1}$ are admissible.

The general structure of our enumeration algorithm is a  depth-first-left-to-right (DFLR) traversal over all admissible stacks $s$. For this, it calls the recursive procedure {\sf extend} (shown as Algorithm~\ref{mainEnumAlgoNonCompressed}) with the initial admissible stack $s = \varepsilon$. Note that whenever {\sf extend}$(s)$ is called, $|s| < m$ holds. It is clear that the call {\sf extend}$(\varepsilon)$ triggers the enumeration of all admissible stacks $a_1 a_2 \cdots a_m$. In an implementation one would store $s$ as a global variable.

\begin{algorithm}[t]
\SetKwComment{Comment}{(}{)}
\SetKwInput{KwGlobal}{variables}
\SetKw{KwOutput}{output}
\SetKw{KwReturn}{return}
\SetKwFor{For}{for all}{do}

$\ell := |s|+1$\;
\For{$a \in L_\ell$ such that $sa$ is admissible}{
\textbf{if} $\ell = m$ \textbf{then} \KwOutput $sa$ \textbf{else} {\sf extend}$(sa)$}
\KwReturn
\caption{$\mathsf{extend}(s)$  \label{mainEnumAlgoNonCompressed}}
\end{algorithm}

Let us assume that we can check whether a stack $s$ is admissible in time $f(d, |\phi|)$ (it is not hard to see that this is possible, and this aspect will anyway be discussed in detail for the compressed setting in Section~\ref{sec-enum-compressed}). After the initial call {\sf extend}$(\varepsilon)$,
the algorithm  constructs an admissible stack $s$ with $|s| = q$ (or terminates) after time bounded in $d, k, r$ and $\bigR$ (since the lists $L_{1}, \ldots, L_{q}$ are short). If some $a \in L_{q + 1}$ is \emph{non-admissible}, i.e., the stack $sa$ is not admissible, then $\dist_{\RS}(t_{a_i,\sigma_i}, t_{a,\sigma_{q+1}}) \leq 2r+1$ and therefore $\dist(a_i, a) \leq 2\bigR+2r+1$ for some $i \in [q]$. Thus, the total number of non-admissible elements from $L_{q + 1}$ can be bounded by a function of $d, k, r$ and $\bigR$. Consequently, since $L_{q+1}$ is long, the algorithm necessarily finds some admissible $a \in L_{q + 1}$ (or terminates) after time bounded in $d, k, r$ and $\bigR$. From this observation, the following lemma can be concluded with moderate effort; see Appendix~\ref{sec:delayBoundAppendix}.

\begin{lemma}\label{delayBoundLemma}
The delay of the above enumeration procedure is bounded by $f(d, |\phi|)$.
\end{lemma}

\section{Straight-Line Programs for Relational Structures}\label{S hier}

In this section, we introduce the compression scheme that shall be used to compress relational structures. We first need the definition of pointed structures.

For $n \geq 0$, an \emph{$n$-pointed structure}
is a pair $(\RS,\tau)$, where $\RS$ is 
a structure and $\tau : [n] \to \RS$
is injective. The elements in $\ran(\tau)$ ($\RS \setminus \ran(\tau)$, respectively) 
are called {\em contact nodes} (\emph{internal nodes}, respectively).
The node $\tau(i)$ is called the \emph{$i$-th contact node}.

A \emph{relational straight-line program} (\emph{r-SLP} or just \emph{SLP})
is a tuple $D=(\mathcal R, N, S, P)$, where
\begin{itemize}
\item $\mathcal R$ is a relational signature,
\item $N$ is a finite set of \emph{nonterminals}, where 
every $A \in N$ has a {\em rank} $\rank(A) \in \mathbb{N}$,
\item $S \in N$ is the \emph{initial nonterminal}, where $\rank(S) = 0$, and
\item $P$ is a set of \emph{productions} that contains for every $A \in N$ a unique production 
$A \to (\RS_A,\tau_A,E_A)$ with $(\RS_A, \tau_A)$  a $\rank(A)$-pointed structure over the signature $\mathcal R$ and $E_A$ a multiset of \emph{references} of the form $(B,\sigma)$, where $B \in N$ and $\sigma : [\rank(B)] \to \RS_A$ is injective. 
\item Define the binary relation $\to_D$ on $N$ 
as follows: $A \to_D B$ if and only if  $E_A$ contains a reference of the form $(B,\sigma)$. 
Then we require that $\to_D$ is acyclic. Its transitive closure $\succ_D$ is 
a partial order that we call the \emph{hierarchical order} of $D$.
\end{itemize}
Let $|D| = \sum_{A \in N} (|\RS_A| +\sum_{(B,\sigma) \in E_A} (1+\rank(B)))$ be the
 size of $D$. We define the
ordered dag $\DAG(D)=(N,\gamma,S)$, 
where the child-function $\gamma$
is defined as follows:
Let $B \in N$ and let $(B_1,\sigma_1),\ldots, (B_m, \sigma_n)$ be an enumeration of the references in $E_B$, where
every reference appears in the enumeration as many times as in the multiset $E_B$.
The specific order of the references is not important and assumed to be somehow given by the input encoding of $D$
We then define $\gamma(B) = B_1 \cdots B_n$. 
 
We now define for every nonterminal $A \in N$ a $\rank(A)$-pointed relational structure $\eval(A)$ (the value of $A$).
We do this on an intuitive level, a formal definition can be found in Appendix~\ref{SLPDefinitionsAppendix}.
 If $E_A = \emptyset$, then we define $\eval(A) = (\RS_A, \tau_A)$. If, on the other hand, $E_A \neq \emptyset$, then 
 $\eval(A)$ is obtained from $(\RS_A, \tau_A)$ by expanding all references in $E_A$.
 A reference $(B, \sigma) \in E_A$ is expanded by the following steps: (i) create the disjoint union of $\RS_A$ and $\RS_B$, (ii) merge node $\tau_B(i) \in \RS_B$ with node $\sigma(i) \in \RS_A$ for every $i \in [\rank(B)]$, (iii) remove $(B, \sigma)$ from $E_A$, and (iv) add all references from $E_B$ to $E_A$. Due to the fact that $\to_D$ is acyclic, we can keep on expanding references (the original ones from $E_A$ and the new ones added by the expansion operation) in any order until there are no references left. The resulting relational structure is $\eval(A)$; see Example~\ref{ex:SLP} and Figure~\ref{fig:SLP} for an illustration.
 
We define $\eval(D) = \eval(S)$. Since $\rank(S) = 0$ it can be viewed as an ordinary ($0$-pointed) structure. It is not hard to see that $|\eval(D)| \le 2^{\bigo(|D|)}$ and that this upper bound can be also reached. Thus, $D$ can be seen as a compressed representation of the structure $\eval(D)$. 

In Section~\ref{sec-FOL} we claimed that FO-query enumeration can be reduced to the case where $\mathcal{R}$ only contains relation symbols of arity at most two (with the details given in Appendix~\ref{sec:generalDefinitionsAppendix}). 
It is easy to see that this reduction can be also done in the SLP-compressed setting simply by applying the reduction to all structures $\RS_A$;
see Appendix~\ref{SLPDefinitionsAppendix} for details.

We say that the SLP $D=(\mathcal R, N, S, P)$ is 
\emph{apex}, if for every $A \in N$ and every reference $(B,\sigma) \in E_A$ we have $\ran(\sigma) \cap \ran(\tau_A) = \emptyset$. Thus, contact nodes of a right-hand side cannot be accessed by references.
Apex SLPs are called $1$-level restricted in \cite{MHR94}. 
It is easy to compute the maximal degree of nodes in $\mathcal{G}(\eval(D))$ for an apex SLP $D$: 
for every node $v$ in a structure $\RS_A$ compute $d_v$ as the sum of (i) the degree of $v$ in $\mathcal{G}(\RS_A)$ and
(ii) for every reference $(B,\sigma) \in E_A$ and every $i \in [\rank(B)]$ with $v = \sigma(i)$, the degree of $\tau_B(i)$ in $\mathcal{G}(\RS_B)$.
Then the maximum of all these numbers $d_v$ is the maximal degree of nodes in $\mathcal{G}(\eval(D))$.
The apex property implies a certain locality property for $\eval(D)$ that will be explained in Section~\ref{sec node reps}.
In the rest of the paper we will mainly consider apex SLPs.

A simple example of a class of graphs that are exponentially compressible with apex SLPs is the class of perfect binary trees.
The perfect binary tree of height $n$ (with $2^n$ leaves) can be produced by an apex SLP of size $\mathcal{O}(n)$. 
Here is an explicit example for an apex SLP:

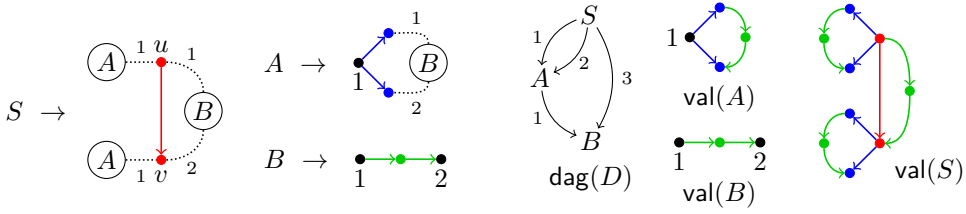
\begin{figure}
\tikzset{circlenode/.style={draw, circle, inner sep = 1.2pt, minimum size = 14pt}}
\tikzset{alpha/.style={inner sep = 1.2pt}}
\tikzset{bullet/.style={draw, circle, fill, inner sep = 1.2pt}}
\centering{
\begin{tikzpicture}[label distance=-.5mm]
\node[alpha] (S) {$S \ \to$};
\node[circlenode, above right = .3cm and .3cm of S] (A1) {$A$};
\node[circlenode, below right = .3cm and .3cm of S] (A2) {$A$};
\node[bullet, fill = red, draw = red, right = .4cm of A1, label=above:$u$] (t1) {};
\node[bullet, fill = red, draw = red, right = .4cm of A2, label=below:$v$] (t2) {};
\draw[->, semithick, red] (t1) to (t2);
\node[circlenode, right = 1.5cm of S] (B1) {$B$};
\draw[densely dotted,semithick] (A1) to node[midway, above] {\scriptsize $1$} (t1);
\draw[densely dotted,semithick] (A2) to node[midway, below] {\scriptsize $1$} (t2);
\draw[densely dotted,semithick] (B1) to[out=90,in=0] node[midway, above] {\scriptsize $1$} (t1);
\draw[densely dotted,semithick] (B1) to[out=-90,in=0] node[midway, below] {\scriptsize $2$} (t2);

\node[alpha, above right = .3cm and 2.5cm of S] (A) {$A \ \to$};
\node[bullet, right = .3cm of A, label=below:$1$] (t6) {};
\node[bullet, fill = blue, draw = blue, above right =  0.3cm and .3cm of t6] (t7) {};
\node[bullet, fill = blue, draw = blue, below right =  0.3cm and .3cm of t6] (t8) {};
\draw[->, blue, semithick] (t6) to (t7);
\draw[->, blue, semithick] (t6) to (t8);
\node[circlenode, right = .6cm of t6] (B2) {$B$};
\draw[densely dotted,semithick] (B2) to[out=90,in=0] node[midway, above] {\scriptsize $1$} (t7);
\draw[densely dotted,semithick] (B2) to[out=-90,in=0] node[midway, below] {\scriptsize $2$} (t8);

\node[alpha, below right = .3cm and 2.5cm of S] (B) {$B \ \to$};
\node[bullet, right = .3cm of B, label=below:$1$] (t3) {};
\node[bullet,  fill = darkgreen, draw = darkgreen, right = .4cm of t3] (t4) {};
\node[bullet, right = .4cm of t4, label=below:$2$] (t5) {};
\draw[->, semithick, darkgreen] (t3) to (t4);
\draw[->, semithick, darkgreen] (t4) to (t5);

\node[alpha, above right = .1cm and 6cm of S] (A') {$A$};
\node[alpha, above right = .5cm and .3cm of A'] (S') {$S$};
\node[alpha, below right = .5cm and .3cm of A'] (B') {$B$};
\node[alpha, below = .1cm of B'] {$\DAG(D)$};
\draw[->, bend left] (S') to node[midway, right] {\scriptsize $3$} (B');
\draw[->, bend left] (S') to node[pos = .7, right] {\scriptsize $2$} (A');
\draw[->, bend right] (S') to node[midway, left] {\scriptsize $1$} (A');
\draw[->, bend right] (A') to node[midway, left] {\scriptsize $1$} (B');

\node[bullet,  fill = darkgreen, draw = darkgreen, below right = .2cm and 8.5cm of S] (t4) {};
\node[bullet, left = .4cm of t4, label=below:$1$] (t3) {};
\node[bullet, right = .4cm of t4, label=below:$2$] (t5) {};
\draw[->, semithick, darkgreen] (t3) to (t4);
\draw[->, semithick, darkgreen] (t4) to (t5);
\node[alpha, below = .4cm of t4] {$\eval(B)$};

\node[bullet, above  left =  1.3cm and .3cm of t4, label=left:$1$] (n1) {};
\node[bullet, fill = blue, draw = blue, above right =  0.3cm and .3cm of n1] (n2) {};
\node[bullet, fill = blue, draw = blue, below right =  0.3cm and .3cm of n1] (n3) {};
\draw[->, blue, semithick] (n1) to (n2);
\draw[->, blue, semithick] (n1) to (n3);
\node[bullet, fill = darkgreen, draw = darkgreen, right = .6cm of n1] (n4) {};
\draw[->,darkgreen, semithick] (n2) to[out=0,in=90]  (n4);
\draw[->,darkgreen, semithick] (n4) to[out=-90,in=0] (n3);
\node[alpha, below = .12cm of n3] {$\eval(A)$};

\node[bullet, darkgreen, above right = .5mm and 11cm of S] (v1) {};
\node[bullet, fill = red, draw = red, above left =  0.6cm and .3cm of v1] (v2) {};
\node[bullet, fill = red, draw = red, below left =  0.6cm and .3cm of v1] (v3) {};
\draw[->,darkgreen, semithick] (v2) to[out=0,in=90]  (v1);
\draw[->,darkgreen, semithick] (v1) to[out=-90,in=0] (v3);
\draw[->, red, semithick] (v2) to (v3);
\node[alpha, below right = .1cm and .1cm of v3] {$\eval(S)$};

\node[bullet, fill = blue, draw = blue, above left =  0.3cm and .3cm of v2] (v4) {};
\node[bullet, fill = blue, draw = blue, below left =  0.3cm and .3cm of v2] (v5) {};
\draw[->, blue, semithick] (v2) to (v4);
\draw[->, blue, semithick] (v2) to (v5);
\node[bullet, darkgreen, left = .6cm of v2] (v8) {};
\draw[->,darkgreen, semithick] (v4) to[out=180,in=90]  (v8);
\draw[->,darkgreen, semithick] (v8) to[out=-90,in=180] (v5);

\node[bullet, fill = blue, draw = blue, above left =  0.3cm and .3cm of v3] (v6) {};
\node[bullet, fill = blue, draw = blue, below left =  0.3cm and .3cm of v3] (v7) {};
\draw[->, blue, semithick] (v3) to (v6);
\draw[->, blue, semithick] (v3) to (v7);
\node[bullet, darkgreen, left = .6cm of v3] (v9) {};
\draw[->,darkgreen, semithick] (v6) to[out=180,in=90]  (v9);
\draw[->,darkgreen, semithick] (v9) to[out=-90,in=180] (v7);

\end{tikzpicture}}
\caption{The SLP $D$ of Example~\ref{ex:SLP} together with $\DAG(D)$ and  $\eval(X)$
for $X \in \{S,A,B\}$.}
\label{fig:SLP}
\end{figure}

\begin{example}\label{ex:SLP}
Consider the SLP $D=(\mathcal R, N, S, P)$ where $\mathcal R$ only contains a binary relation symbol $r_1$ and $N = \{S, A, B\}$ with $\rank(S) = 0$, $\rank(A) = 1$ and $\rank(B) = 2$. The productions of these nonterminals are depicted on the left of Figure~\ref{fig:SLP}. For instance, the production $S\to (\RS_S,\tau_S,E_S)$ consists of the $0$-pointed structure $(\RS_S,\tau_S)$, where the universe of $\RS_S$ consists of 
the two red nodes $u$ and $v$, and the reference set $E_S = \{(A,\sigma_1), (A,\sigma_2), (B,\sigma_3)\}$ with $\sigma_1(1) = u$, $\sigma_2(1) = v$, $\sigma_3(1) = u$ and $\sigma_3(2) = v$ (in Figure~\ref{fig:SLP} each $\sigma_i(j)$ is connected by a
$j$-labeled dotted line with the nonterminal). The production for nonterminal $B$ consists of a $2$-pointed structure (and no references), the contact nodes of which are labeled by $1$ and $2$. The structure $\eval(D) = \eval(S)$ is shown on the right of Figure~\ref{fig:SLP}. It can be obtained by first constructing $\eval(A)$  by replacing the single $B$-reference in $\RS_A$ by $\RS_B = \eval(B)$. Note that 
$1$- and $2$-labeled dotted lines identify the two nodes to be merged with the two contact nodes of $\RS_B$, and that $\eval(A)$ has exactly one contact node. Then we replace the  $B$-reference in $\RS_S$ by $\eval(B)$ 
and both $A$-references in $\RS_S$ by $\eval(A)$. This merges $u$ (and $v$) with the contact node of the first (and the second) occurrence of $\eval(A)$. Red (resp., blue, green) edges and nodes are produced from $S$ (resp., $A$, $B$).

Since no contact node is adjacent to any reference, this SLP is apex. The size of $\eval(D)$ is $31$. 
The size of $D$ is $26$: $9$ (for the $S$-production) $+ \ 10$ (for the $A$-production) $+ \ 7$ (for the $B$-production).
\end{example}

\subsection{Representation of nodes of an SLP-compressed structure}
\label{sec node reps}

Let $A \in N$. A node $a \in \eval(A)$ can be uniquely represented by a pair $(p,v)$ 
such that $p$ is an $A$-path in $\DAG(D)$ and one of the following two cases holds:
\begin{itemize}
\item $p$ ends in $B \in N \setminus \{A\}$ and $v \in \RS_B \setminus \ran(\tau_B)$ is an internal node.\footnote{The nodes in $\ran(\tau_B)$,
i.e., the contact nodes of $\RS_B$, are excluded here,
because they were already generated by some larger (with respect to the
hierarchical order $\succ_D$) nonterminal.}
\item $p = A$ and $v \in \RS_A$. 
\end{itemize}
We call this the \emph{$A$-representation of $a$}. The $S$-representations of the nodes of $\eval(S) = \eval(D)$ 
are also called \emph{$D$-representations}. Note that if $(p, v)$ is the $D$-representation of a node then $v \in \RS_A \setminus \ran(\tau_A)$ 
for some $A \in N$ (since $\rank(S) = 0$). We will often identify a node of $\eval(A)$ with its $A$-representation; in particular when $A = S$.
One may view a $D$-representation $(p,v)$  as a stack $pv$. In order to construct outgoing (or incoming) edges of $(p,v)$ in the 
structure $\eval(D)$, one only has to modify this stack at its end; see 
also Appendix~\ref{sec compute expansion}.

The apex condition implies a kind of locality in $\eval(D)$ that can be nicely formulated in terms 
of $D$-representations: 
If two nodes $a = (p,u)$ and $b = (q,v)$ have distance $\zeta$ in the graph $\mathcal{G}(\eval(D))$
then the prefix distance between $p$ and $q$ (which is the number of edges in $p$ and $q$ that do not belong to the longest common prefix of $p$ and $q$)
 is also at most $\zeta$. This property is exploited several times in the paper.

Based on $A$-representations, we can define a natural embedding
of $\eval(B)$ into $\eval(A)$ in case $A \succ_D B$.
Assume that $p$ is a non-empty $A$-to-$B$ path in $\DAG(D)$ with $A \neq B$.
Let us write $p = p' C i B$ for some nonterminal $C$
(we may have $C=A$).
Let $(B, \sigma) \in E_C$ be the unique reference that corresponds to the edge $(C,i,B)$ in $\DAG(D)$.
We then define the embedding
$\eta_{p} : \eval(B) \to \eval(A)$ as follows, where $(q,v)$ is  a node in $\eval(B)$ given by its $B$-representation so that $q$ is a $B$-path
(recall that the path $pq$ is obtained by concatenating the paths $p$ and $q$; see Section~\ref{sec-dag}):
\[
\eta_{p}(q,v) = \begin{cases} 
   (p' C, \sigma(i)) & \text{ if $q = B$ and $v = \tau_B(j)$ for some $j \in [\rank(B)]$,} \\
   (pq,v) & \text{ otherwise.}
   \end{cases}
 \]
 We can extend this definition to the case $A=B$ (where $p=A$) by defining $\eta_p$ as the identity map on $\eval(A) = \eval(B)$.
If $\RS$ is the substructure of $\eval(B)$ induced by the set $U \subseteq \eval(B)$ then we write 
$\eta_p(\RS)$ for the substructure of $\eval(A)$ induced by the set $\eta_p(U)$. Note that in general we do not have
$\eta_p(\RS) \simeq \RS$. For instance, if $\RS = \eval(B)$ then in $\eval(A)$ there can be edges between contact nodes of $\eval(B)$ 
that are generated by a nonterminal $C$ with $C \to_D B$.

Recall the definition of the lexicographic order on the set of all $A$-paths of $\DAG(D)$ for $A \in N$ (see
Section~\ref{sec-dag}). We define $\lex_A(p)$ as the position of $p$ in the lexicographically sorted list of all 
$A$-paths of $\DAG(D)$, where we start with $0$ (i.e., $\lex_A(A) = 0$; note that $A$ is the empty
path starting in $A$ and hence the lexicographically smallest path among all $A$-paths). For 
$\lex_S(p)$ we just write $\lex(p)$.
Later it will be convenient to represent the initial path component $p$ of a $D$-representation $(p,v)$ by the number $\lex(p)$ and call
$(\lex(p), v)$ be the \emph{lex-representation} of the node $a = (p,v) \in \eval(D)$. The number of initial paths in $\DAG(D)$ can be bounded by $2^{\bigo(|D|)}$: the number of initial-to-leaf
 paths in $\DAG(D)$ is bounded by $3^{|\DAG(D)|/3} \leq 3^{|D|/3}$ (this is implicitly shown in the proof of \cite[Lemma~1]{CLLLPPSS05})
 and the number of all initial paths in $D$ is bounded by twice  the number of initial-to-leaf paths in $D$. 
 Hence, the numbers $\lex(p)$ have bit length $\bigo(|D|)$.

\begin{example}\label{ex:embed}
Recall the SLP $D$ from Example~\ref{ex:SLP} and $\DAG(D)$ shown to the right of $D$'s productions in Figure~\ref{fig:SLP}.
Then the pairs $(S,u)$ and $(S, v)$ (recall that $u$ and $v$ are the two nodes of $\RS_S$) represent the two red nodes of $\eval(D) = \eval(S)$,
and $(S3B, w)$, where $w$ is the green node in $\RS_B$, represents the rightmost green node of $\eval(D)$.
Its lex-representation is $(5,w)$ (there are six initial paths in $\DAG(D)$).
 As another example, the two leftmost (green) nodes of $\eval(D)$ are represented by the pairs $(S1A1B,w)$ and $(S2A1B,w)$
 with the lex-representations $(2,w)$ and $(4,w)$, respectively. 
 For the $S$-to-$B$ path $p = S2A1B$ in $\DAG(D)$ we have $\eta_{p}(B, w) = (S2A1B, w)$ 
  and $\eta_{p}(B, \tau_B(1)) = (S2A, \sigma(1))$, where $(B, \sigma)$ is the only reference in $E_{A}$.
\end{example}

\subsection{Register length in the compressed setting} \label{sec:SLPRAM}

In the following sections we will develop an enumeration algorithm for the set of all tuples in $\phi(\eval(D))$, where the SLP $D$ is part
of the input. Recall that $\eval(D)$ may contain $2^{\Theta(|D|)}$ many elements. In order to achieve constant delay, we therefore should
set the register length in our algorithm to $\Theta(|D|)$ so that we can store elements of $\eval(D)$. This is in fact a standard assumption 
for algorithms on SLP-compressed objects. For instance, when dealing with SLP-compressed strings, one usually assumes that registers can store
positions in the decompressed string. We only allow additions, subtractions and comparisons on these $\Theta(|D|)$-bit registers and these
operations take constant time (since we assume the uniform cost measure).
For  registers of length $\bigo(\log |D|)$ we will also allow pointer operations. 

Note that a $D$-representation $(p, v)$ needs $\bigo(|D|)$ many $\bigo(\log |D|)$-bit registers, whereas its lex-representation
$(\lex(p), v)$ fits into two registers (one of length $\bigo(\log |D|)$).

\section{FO-Enumeration over SLP-Compressed Degree-Bounded Structures}
\label{sec-enum-compressed}

We now have all definitions available in order to state a more precise version of Theorem~\ref{prelimMainResultTheorem}:

\begin{theorem}\label{prelimMainResultTheorem-2}
Given an apex SLP $D$ such that $\eval(D)$ is degree-$d$ bounded
and an FO-formula $\phi(x_1, \ldots, x_k)$, we can enumerate the result set $\phi(\RS)$ with preprocessing time $f(d,|\phi|) \cdot |D|$ and delay $f(d,|\phi|)$ for some computable function $f$. All nodes of $\phi(\RS)$ are output in their lex-representation.
\end{theorem}
Throughout Section~\ref{sec-enum-compressed} we fix $D = (\mathcal R, N, S, P)$ 
and $\phi(x_1, \ldots, x_n)$ as in Theorem~\ref{prelimMainResultTheorem-2}.
Let $\mathsf{qr}(\phi) = \nu$. W.l.o.g.~we can assume that $d \ge 2$.

The general structure of our enumeration algorithm is the same as for the uncompressed setting. In particular, we also use Gaifman-locality to reduce to the problem of enumerating for a fixed $\mathcal{B} \in \mathcal{T}_{k,r}$ the set of all $\mathcal{B}$-tuples $\bar{a} \in \eval(D)^k$  (see Appendix~\ref{sec-gaifman}), which then reduces to the problem of enumerating for all consistent factorizations $\Lambda = (\mathcal{B}_1, \sigma_1,  \ldots, \mathcal{B}_m, \sigma_m)$ of $\mathcal{B}$ the set of all $m$-tuples $\bar{b} \in \eval(D)^m$ that are admissible for $\Lambda$ (see the beginning of Section~\ref{sec:uncompressed}). 

Here, a first complication occurs: one important component of the above reduction for the uncompressed setting is that FO model checking on degree-$d$ bounded structures can be done in time $|\RS| \cdot f(d,|\phi|)$~\cite{Seese96}. For the SLP-compressed setting we do not have a linear time (i.\,e., in time $|D| \cdot f(d,|\phi|)$) model checking algorithm. Only an {\sf NL}-algorithm for apex SLPs is known \cite{Lohrey2012}. It is not hard to obtain a linear time algorithm from the {\sf NL}-algorithm in \cite{Lohrey2012}, but there is an easier solution
 that bypasses model checking. We give all the details of how to perform the necessary reduction in the compressed setting in Appendix~\ref{sec:GaifmanReductionCompressed}.

Consequently, as in the uncompressed setting, it suffices to consider a fixed consistent factorization 
$\Lambda = (\mathcal{B}_1, \sigma_1, \ldots, \mathcal{B}_m, \sigma_m)$
of $\mathcal{B}$
and to enumerate the set of all $m$-tuples in $\eval(D)$ that are admissible for $\Lambda$.
As before we define the larger radius $\bigR = 2rk - r + k -1$; see \eqref{def-bigR}.

\subsection{Expansions of nonterminals}\label{sec:expansions}

In this section we introduce the concept of $\zeta$-expansions for a constant $\zeta \geq 1$ (later, $\zeta$ will be a constant of the form $f(d,|\phi|)$), which will be needed to transfer the enumeration algorithm for the uncompressed setting (Section~\ref{sec:enumAlgo}) to the SLP-compressed setting. The idea is to apply the productions from $D$, starting with a nonterminal $A \in N$, until all nodes of $\eval(A)$ that have distance at most $\zeta$
from the nodes in the right-hand side of $A$ (except for the contact nodes of $A$) are produced. 
For a nonterminal $A \in N$ we define
\[
\mathsf{In}_A  =  \{ (A, v) : v \in \RS_A \setminus \ran(\tau_A) \} \subseteq \eval(A).
\]
These are the internal nodes of $\eval(A)$ (written in $A$-representation)  that are directly produced with the production 
$A \to (\RS_A, \tau_A, E_A)$.  Let $a_1, \ldots, a_m$ be a list of all nodes from $\mathsf{In}_A$. We then define
the \emph{$\zeta$-expansion} as the following induced substructure of $\eval(A)$:
\[ \expan_{\zeta}(A) = \neighborhoodStr{\eval(A)}{\zeta}{a_1, \ldots, a_m} . \]
We always assume that the nodes of $\expan_{\zeta}(A)$ are represented by their $A$-representations.
Let
\[
\mathsf{Bd}_{A,\zeta}  =  \{ (A, v) : v \in \ran(\tau_A) \} \cup \{ a \in \eval(A) : \dist_{\eval(A)}(\mathsf{In}_A, a) = \zeta \} \subseteq \eval(A)
\] 
be the \emph{boundary} of $\expan_{\zeta}(A)$.
A \emph{valid substructure} of $\expan_{\zeta}(A)$ is an induced substructure
$\mathcal{A}$ of $\expan_{\zeta}(A)$ with $\mathcal{A} \cap \mathsf{Bd}_{A,\zeta} = \emptyset \neq \mathcal{A} \cap \mathsf{In}_A$.
If $\mathcal{A}$ is a valid substructure of $\expan_{\zeta}(A)$
and $p$ is an $S$-to-$A$ path in $\DAG(D)$, then  any neighbor of $\eta_p(\mathcal{A})$ in the graph $\mathcal{G}(\eval(D))$ belongs to $\eta_p(\expan_{\zeta}(A))$. Moreover, $\eta_p(\mathcal{A}) \simeq \mathcal{A}$, since all contact nodes $(A, \tau_A(i))$
are excluded from a valid substructure of $\expan_{\zeta}(A)$.
In the following, we consider the radius $\zeta = 2 \bigR+1$. For a nonterminal $A \in N$ we write
$\expan(A)$ for the expansion $\expan_{2 \bigR+1}(A)$ in the rest of the paper. 

Fix a $\bigR$-neighborhood type $\mathcal{B}$.
A node $a \in \expan(A) \subseteq \eval(A)$ is called a \emph{valid $\mathcal{B}$-node} in $\expan(A)$
if (i) $\neighborhoodStr{\expan(A)}{\bigR}{a} \simeq \mathcal{B}$ and (ii) 
$\neighborhoodStr{\expan(A)}{\bigR}{a}$ is a valid substructure of $\expan(A)$.
We say that $A$ is \emph{$\mathcal{B}$-useful} if there is a valid $\mathcal{B}$-node in $\expan(A)$.
We consider now the following two sets:
\begin{itemize}
\item $\mathsf{S}^{\mathcal{B}}_1 = \{ (p,a) : \exists A \in N : \text{ $p$ is an $S$-to-$A$ path in $\DAG(D)$, $a$ is a valid $\mathcal{B}$-node in $\expan(A)$} \}$
\item $\mathsf{S}^{\mathcal{B}}_2 = \{ b \in \eval(D) : b \text{ is a $\mathcal{B}$-node} \}$
\end{itemize}
We define a mapping $h : \mathsf{S}^{\mathcal{B}}_1  \to \eval(D)$ as follows. Let $(p,a) \in \mathsf{S}^{\mathcal{B}}_1$, where $p$ is an $S$-to-$A$ path in $\DAG(D)$ and let $(q,v)$ be the $A$-representation of $a \in \expan(A)$. We then define $h(p, a) = \eta_p(a) = (pq, v)$ (where the latter is a $D$-representation that we identify as usual with a node from $\eval(D)$). The proof of the following lemma can be 
found in Appendix~\ref{sec:algoPrelimAppendix}.

\begin{lemma}\label{bijectionLemma}
The mapping $h$ is a bijection from  $\mathsf{S}^{\mathcal{B}}_1$ to $\mathsf{S}^{\mathcal{B}}_2$.
\end{lemma}

\subsection{Overview of the enumeration algorithm}\label{sec:generalStructure}

Our goal is to carry out the algorithm described in Section~\ref{sec:enumAlgo}, but in the compressed setting, i.e., 
by only using the apex SLP $D = (\mathcal{R},N,S, P)$ instead of the explicit structure $\eval(D)$.  
As in the uncompressed setting, it suffices to consider a fixed $(k,r)$-neighborhood type $\mathcal{B} \in \mathcal{T}_{k,r}$ together with a 
fixed consistent factorization 
\begin{equation}
\label{eq-consistent-factorization}
\Lambda = (\mathcal{B}_1, \sigma_1, \ldots, \mathcal{B}_m, \sigma_m)
\end{equation}
  of $\mathcal{B}$ and to enumerate the set of all $m$-tuples in $\eval(D)$ that are admissible for $\Lambda$.
In the following we sketch the algorithm; details can be found in Appendix~\ref{sec:algoPrelimAppendix}.

\subparagraph{Enumeration of all \boldmath{$\mathcal{B}_i$}-nodes.}
The algorithm for the uncompressed setting (Section~\ref{sec:uncompressed}) precomputes 
for every $\mathcal{B}_i$ 
a list $L_i$ of all $\mathcal{B}_i$-nodes of the structure $\RS$.
This is no longer possible in the compressed setting since the structure $\eval(D)$ is too big.
However, as shown in Section~\ref{sec:expansions}, there
is a bijection between the set of  $\mathcal{B}_i$-nodes in $\eval(D)$ and the set
of all pairs $(p, a)$, where $p$ is an $S$-to-$A$ path in $\DAG(D)$ for a 
$\mathcal{B}_i$-useful nonterminal $A$
and $a$ is a valid $\mathcal{B}_i$-node in $\expan(A)$ that is written in its $A$-representation $(q,v)$.
Hence, on a high level, instead of explicitly precomputing the lists $L_i$ of all $\mathcal{B}_i$-nodes, we enumerate them with Algorithm~\ref{B-enumeration compressed}.

\begin{algorithm}[t]
\SetKwComment{Comment}{(}{)}
\SetKwInput{KwGlobal}{variables}
\SetKw{KwOutput}{output}
\SetKw{KwReturn}{return}
\SetKwFor{For}{for all}{do}

\For{initial paths $p$ in $\DAG(D)$ that end in a $\mathcal{B}_i$-useful nonterminal $A$}{
\For{$(q,v) \in \expan(A)$ that are valid $\mathcal{B}_i$-nodes in $\expan(A)$}{
\Return $(\lex(p),q,v)$}}
\caption{enumeration of all $\mathcal{B}_i$-nodes \label{B-enumeration compressed}}
\end{algorithm}

To execute this algorithm we first have to compute in the preprocessing all expansions $\expan(A)$ for a nonterminal $A$.
This is easy: using a breath-first-search (BFS), 
we locally generate $\eval(A)$ starting with the nodes in $\mathsf{In}_A$ until all nodes $a \in \eval(A)$ with
$\dist_{\eval(A)}(\mathsf{In}_A, a) \leq 2\bigR+1$ are generated.  The details can be found in Appendix~\ref{sec compute expansion}.
The size of $\expan(A)$ is bounded by $|\RS_A| \cdot f(d, |\phi|)$
(the size of a $(2\bigR+1)$-sphere around a tuple of length at most $|\RS_A|$ in a degree-$d$ bounded structure)
and can be constructed in time $|\RS_A| \cdot f(d, |\phi|)$.
Summing over all $A \in N$ shows that all $(2\bigR+1)$-expansions can be precomputed in time $|D| \cdot f(d, |\phi|)$.

With the $\expan(A)$ available, we can easily precompute brute-force the set of all valid $\mathcal{B}_i$-nodes in $\expan(A)$ (needed in Line 2 of 
Algorithm~\ref{B-enumeration compressed})
and then the set of all $\mathcal{B}_i$-useful nonterminals 
(needed in Line 1 of Algorithm~\ref{B-enumeration compressed}). 
Recall that $A$ is $\mathcal{B}_i$-useful iff there is a valid $\mathcal{B}_i$-node in $\expan(A)$.
Moreover, for every valid $\mathcal{B}_i$-node $c = (q,v) \in \expan(A)$
we compute also an isomorphism $\pi_c : \mathcal{B}_i \to \neighborhoodStr{\expan(A)}{\bigR}{c_i}$.
The time for this is bounded by $f(d,|\varphi|)$ for one nonterminal $A$ and hence by 
$|D| \cdot f(d, |\phi|)$ in total. More details on this part can be found in Appendix~\ref{sec:computeneighborhoodTypes}.

The most challenging part of Algorithm~\ref{B-enumeration compressed} is the enumeration of all initial paths $p$ in $\DAG(D)$
that end in a $\mathcal{B}_i$-useful nonterminal (Line 1). Let $\mathcal{P}_i$ be the set of these paths.
 In constant delay, we cannot afford to output a path $p \in \mathcal{P}_i$ as a list of edges (it does not fit into a constant number of registers in our machine model,
 see Section~\ref{sec:SLPRAM}). That is why we return the number $\lex(p)$ (which fits into a single register in our machine model) in Line 3. The idea for constant-delay path enumeration is to run over all paths $p \in \mathcal{P}_i$ in lexicographical order and thereby maintain the number $\lex(p)$.
 The path $p$ is internally stored in a contracted form. If $\DAG(D)$ would be a binary dag, then we could
 use an enumeration algorithm from \cite{LohreyMR18}, where maximal subpaths of left (right, respectively) outgoing edges are contracted to single edges.
 In our setting, $\DAG(D)$ is not a binary dag, therefore we have to adapt the technique from \cite{LohreyMR18}. Details can be found in Appendix~\ref{sec:pathEnumPreprocessing}
 (which deals with the necessary preprocessing) and Appendix~\ref{sec:MainPathEnumSection} (which deals with the actual path enumeration).

In order to see how Algorithm~\ref{B-enumeration compressed} can be used to replace the precomputed lists $L_i$ in Algorithm~\ref{mainEnumAlgoNonCompressed} for the uncompressed setting, a few additional points have to be clarified.

\subparagraph{Producing the final output tuples.}
Note that for each enumerated $\mathcal{B}_i$-node $b_i \in \eval(D)$ we have to produce the partial $k$-tuple $t_{b_i, \sigma_i}$
(then the final output tuple is $t_{b_1,\sigma_1} \sqcup t_{b_2,\sigma_2} \sqcup \cdots \sqcup t_{b_m,\sigma_m}$). 
Let us first recall that in the uncompressed setting each partial $k$-tuple $t_{b_i,\sigma_i}$ is defined by $t_{b_i,\sigma_i}(j) = \pi_{b_i}(\sigma_i(j))$ for all $j \in \dom(\sigma_i)$, where $\pi_{b_i} : \mathcal{B}_i \to \neighborhoodStr{\RS}{\bigR}{b_i}$ is a precomputed isomorphism.
In the compressed setting, Algorithm~\ref{B-enumeration compressed} outputs every $\mathcal{B}_i$-node $b_i \in \eval(D)$
as a triple $(\lex(p_i), q_i, v)$, where the initial path $p_i \in \mathcal{P}_i$ ends in some $\mathcal{B}_i$-useful nonterminal $A_i \in N$
and $c_i := (q_i,v_i)$ is a valid $\mathcal{B}_i$-node in $\expan(A_i)$. 
Moreover, we have a  precomputed isomorphism $\pi_{c_i} : \mathcal{B}_i \to \neighborhoodStr{\expan(A_i)}{\bigR}{c_i}$, which yields
the isomorphism $\pi_{b_i} = \eta_{p_i} \circ \pi_{c_i} : \mathcal{B}_i \to \neighborhoodStr{\eval(D)}{\bigR}{b_i}$.
Then, for every $j \in \dom(\sigma_i)$ we can easily compute the lex-representation of $\pi_{b_i}(\sigma_i(j))$. 
We first compute $\pi_{c_i}(\sigma_i(j))$ in its $A_i$-representation $(q_{i,j}, v_{i,j})$
using the precomputed mapping $\pi_{c_i}$. Then the lex-representation of 
$t_{b_i, \sigma_i}(j) =  \pi_{b_i}(\sigma_i(j))$ is $(\lex(p_i q_{i,j}), v_{i,j})$, where $\lex(p_iq_{i,j}) = \lex(p_i)+\lex_{A_i}(q_{i,j})$. Here, $\lex(p_i)$ is produced by Algorithm~\ref{B-enumeration compressed}.
The path $q_{i,j}$ has length at most $2 \bigR+1$ (this is a consequence of the apex condition for $D$).
Its $\lex$-number $\lex_{A_i}(q_{i,j})$ can be computed by summing at most $2 \bigR+1$ many edge weights that were computed in the preprocessing phase
(see Appendix~\ref{sec:pathEnumPreprocessing}).

 \subparagraph{Count total number of \boldmath{$\bigR$}-neighborhoods.}
 In Section~\ref{sec:enumAlgo} we distinguish between short and long lists $L_i$.
 Since in our compressed setting, Algorithm~\ref{B-enumeration compressed} replaces the precomputed list $L_i$
 we have to count the number of triples produced by Algorithm~\ref{B-enumeration compressed} (of course, before we 
 run the algorithm) in the preprocessing phase. 
 This is easy: the number of output triples can be computed by summing over all $\mathcal{B}_i$-useful nonterminals $A$
 the product of (i) the number of $S$-to-$A$ paths in $\DAG(D)$ and (ii) the number of valid $\mathcal{B}_i$-nodes in $\expan(A)$. The latter can be computed in the preprocessing phase.
 Computing the number of $S$-to-$A$ paths (for all $A \in N$) involves a top-down pass (starting in $S$) over $\DAG(D)$ with $|\DAG(D)| \leq |D|$ many additions on $\bigo(|D|)$-bit numbers in total. See Appendix~\ref{appendix-counting} for details.

 \subparagraph{Checking distance constraints.} 
 Recall that we fixed the consistent factorization $\Lambda$ from \eqref{eq-consistent-factorization}
 of the fixed $(k,r)$-neighbor\-hood type $\mathcal{B}$ 
 and want to enumerate all tuples $(b_1, \ldots, b_m) \in \eval(D)^m$ that are admissible for $\Lambda$. The definition
 of an admissible tuple also requires to check whether 
 $\dist_{\eval(D)}(t_{b_{i}, \sigma_{i}}, t_{b_{j}, \sigma_{j}}) > 2r+1$
for all $i \neq j$ (see \eqref{disjointness}). 
The nodes $b_i$ are enumerated with Algorithm~\ref{B-enumeration compressed}, hence  the 
following assumptions hold for all $i \in [m]$:
\begin{itemize}
\item $b_i$ is given by a triple $(\lex(p_i), q_i, v_i)$, 
\item $p_i$ is an initial-to-$A_i$ path in $\DAG(D)$ (for some $\mathcal{B}_i$-useful nonterminal $A_i$), and
\item $c_i := (q_i,v_i)$ is a node (written in $A_i$-representation) from $\expan(A_i)$ such that $c_i$ has $\bigR$-neighborhood type $\mathcal{B}_i$ in $\expan(A_i)$ and $\neighborhoodStr{\expan(A_i)}{\bigR}{c_i}$ is a valid substructure of $\expan(A_i)$.
\end{itemize}
In a first step, we show that if $\dist_{\eval(D)}(t_{b_{i}, \sigma_{i}}, t_{b_{j}, \sigma_{j}}) \leq 2r+1$ then 
there is a path $q$ of length at most $3\bigR-r$ such that $p_i = p_j q$ or $p_j = p_i q$.
For this, the apex property for $D$ is important, since it lower bounds the distance between two nodes $a = (p,u)$ and $a' = (p', v')$
of $\eval(D)$ by the prefix distance between the paths $p$ and $p'$ (i.e., the total number of edges that do not belong to the longest
common prefix of $p$ and $p'$); see Appendix~\ref{sec:disjointnessCheck} for details.

We then proceed in two steps: We first check in time $f(d, |\phi|)$ whether $p_j = p_i q$ or $p_i = p_j q$ for some path $q$ of length at most $3\bigR-r$. For checking $p_j = p_i q$ (the case $p_i = p_j q$ is analogous) we check whether $p_j = p_i$ (by checking $\lex(p_j) = \lex(p_i)$) and if this is not the case, we repeatedly remove the last edge of $p_j$ (for at most $3 \bigR - r$ times) and check whether the resulting path equals $p_i$. However, the whole procedure is complicated by the fact that $p_i$ and $p_j$ are given in a contracted form, where some subpaths are contracted to single edges (see the above paragraph on the path enumeration algorithm for $\DAG(D)$).

In the second step we have to check in time $f(d, |\phi|)$ whether $\dist_{\eval(D)}(t_{b_{i}, \sigma_{i}}, t_{b_{j}, \sigma_{j}}) \leq 2r+1$, assuming that 
$p_j = p_i q$ for some path $q$ of length at most $3\bigR-r$. This boils down to checking, for every $b \in \ran(t_{b_{i}, \sigma_{i}})$ and $b' \in \ran(t_{b_{j}, \sigma_{j}})$, whether $\dist_{\eval(D)}(b, b') \le 2r+1$, which is the case iff $\dist_{\eval(A_i)}(c,\eta_q(c')) \le 2r+1$, where $c, c' \in \eval(A_i)$ correspond to $b, b'$ in the sense that $\eta_{p_i}(c) = b$ and $\eta_{p_j}(c') = b'$. For this we locally construct $\neighborhoodStr{\eval(A_i)}{2r+1}{c}$ by starting a BFS in $c$ and then computing all elements of $\eval(A_i)$ with distance at most $2r + 1$ from $c$ just like we constructed the expansions in Appendix~\ref{sec compute expansion}. Details can be found in Appendix~\ref{sec:disjointnessCheck}.
This concludes our proof sketch for Theorem~\ref{prelimMainResultTheorem-2}.

\section{Conclusions and Outlook}

We presented an enumeration algorithm for FO-queries on structures that are represented succinctly by  apex SLPs.
Assuming that the formula is fixed and the degree of the structure is bounded by a constant, the preprocessing time of our
algorithm is linear and the delay is constant. 

There are several possible directions into which our result can be extended.
One option is to use more general formalisms for graph compression. Our SLPs are based on Courcelle's HR (hyperedge
replacement) algebra, which it tightly related to tree width \cite[Section~2.3]{CouEngel12}. Our SLPs can be viewed as dag-compressed 
expressions in the HR algebra,
where the leaves can be arbitrary pointed structures; see \cite{Lohrey2012} for more details. Another (and in some sense more general)
graph algebra is the VR algebra, which is tightly related to clique width \cite[Section~2.5]{CouEngel12}. It is straightforward to define a notion of 
SLPs based on the VR algebra and this leads to the question whether our result also holds for the resulting VR-algebra-SLPs.

Another interesting question is to what extend the results on enumeration for conjunctive queries \cite{BaganEtAl2007, BerkholzEtAl2020} 
can be extended to the compressed setting. 
In this context, it is interesting to note that model checking for a fixed existential FO-formula on
SLP-compressed structures (without the apex restriction) belongs to {\sf NL}.
It would be interesting to see, whether the constant delay enumeration algorithm from \cite{BaganEtAl2007} for 
free-connex acyclic conjunctive queries can be extended to SLP-compressed structures.

Finally, one may ask whether in our main result (Theorem~\ref{prelimMainResultTheorem-2}) 
the apex restriction is really needed. More precisely, consider
an SLP $D$ such that $\eval(D)$ has degree $d$. Is it possible to construct from $D$ in time $|D| \cdot f(d)$ 
an equivalent apex SLP $D'$ of size $|D| \cdot f(d)$ for a computable function $f$? If this is true then one could
enforce the apex property in the preprocessing. In \cite{EngelfrietHL94} it shown that a set of graphs of bounded degree $d$ that
can be produced by a hyperedge replacement grammar (HRG) $H$ can be also produced by an apex HRG,
but the size blow-up is not analyzed with respect to the parameter $d$ and the size of $H$.


\appendix

\section{Omitted Details from Section~\ref{sec-FOL}}\label{sec:generalDefinitionsAppendix}

We give some more details for first order logic. Let us fix a relational signature $\mathcal R  = \{ R_i: i \in I\}$. Atomic FO-formulas over the signature $\mathcal R$ are of the form $x=y$ and $R(x_1, \ldots, x_n)$, where $R \in \mathcal R$ has arity $n$ and $x, y, x_1, \ldots, x_n$ are first-order variables ranging over elements of the universe. From these atomic FO-formulas we construct arbitrary FO-formulas over the signature $\mathcal R$ using Boolean connectives ($\neg$, $\wedge$, $\vee$) and the (first-order) quantifiers $\exists x$ and $\forall x$ (for a variable $x$ ranging over elements of the universe).

As a general convention, we also write $\psi(x_1, \ldots, x_k)$ to denote that $\psi$ is an FO-formula with free variables $x_1, \ldots, x_k$.
A {\em sentence} is an FO-formula without free variables.
The \emph{quantifier rank} $\mathsf{qr}(\psi)$ of an FO-formula $\psi$ is inductively defined as follows:
$\mathsf{qr}(\psi) = 0$ if $\psi$ contains no quantifiers, $\mathsf{qr}(\neg \psi) = \mathsf{qr}(\psi)$, 
$\mathsf{qr}(\psi_1 \wedge \psi_2) = \mathsf{qr}(\psi_1 \vee \psi_2) = \max\{ \mathsf{qr}(\psi_1), \mathsf{qr}(\psi_2) \}$ and
$\mathsf{qr}(\forall x \psi) = \mathsf{qr}(\exists x \psi) = 1 +  \mathsf{qr}(\psi)$.

The model checking problem (and also the enumeration problem) for first-order logic over arbitrary signatures can be reduced to the case, where all arities are at most two as follows. Let us consider a structure $\RS = (U, (R_i)_{i \in I})$ over an arbitrary relational signature $\mathcal R = \{ R_i : i \in I \}$, where the
arity of $R_i$ is $\alpha_i$.
Let $n = \max \{ \alpha_i : i \in I \}$.
We then define a new signature $\mathcal R'$ consisting of the following relation symbols:
\begin{itemize}
\item a unary  symbol $R_U$,
\item all $R_i$ with $\alpha_i = 1$,
\item a unary  symbol $R'_i$ for every $i \in I$ with $\alpha_i > 1$, and
\item binary  symbols $E_1, \ldots, E_n$.
\end{itemize}
From $\RS$ we construct the new structure $\RS'$ over the signature $\mathcal R'$ by taking the universe\footnote{For sets 
$A_1$ and $A_2$ we write their union as $A_1 \uplus A_2$ if $A_1$ and $A_2$ are disjoint.}
$U \uplus \{ b_{i,\bar a} : i \in I, \bar{a} \in R_i \}$ (the $b_{i,\bar a}$ are new elements) and defining the relations as follows:
\begin{itemize}
\item $R_U = U$,
\item $R_i$ is the same relation as in $\RS$ if $\alpha_i=1$.
\item $R'_i = \{ b_{i,\bar{a}} : \bar{a} \in R_i \}$ for all $i \in I$ with $\alpha_i > 1$, and
\item $E_j =  \{ (b_{i,\bar{a}}, a_j) : i \in I, \alpha_i > 1, \bar{a} = (a_1,\ldots, a_{\alpha_i}) \in R_i \}$ for all $j \in [n]$.
\end{itemize}
Given an FO-formula $\psi(x_1, \ldots, x_k)$ we construct the new FO-formula 
$\psi'(x_1, \ldots, x_k) = \bigwedge_{1 \le j \le k} R_U(x_i)  \wedge \tilde{\psi}(x_1, \ldots, x_k)$
where $\tilde{\psi}(x_1, \ldots, x_k)$ is obtained from $\psi(x_1, \ldots, x_k)$ by restricting all quantifiers
in $\psi$ to elements from $U$ (using the unary predicate $R_U$) and
replacing every atomic subformula
$R_i(y_1, \ldots, y_{\alpha_i})$ with $\alpha_i > 1$ 
by $\exists x : R'_i(x) \wedge \bigwedge_{j=1}^{\alpha_i} E_j(x,y_i)$.
We then have $\psi(\RS) = \psi'(\RS')$. In addition the following facts are important:
\begin{itemize}
\item $\mathsf{qr}(\psi') \le  \mathsf{qr}(\psi)+1$.
\item The degree of  $\RS'$ is bounded by the 
degree of $\RS$.
\item The structure $\RS'$ can be constructed in linear time from $\RS$ and the formula
$\psi'$ can be constructed in linear time from $\psi$.
\end{itemize}

\section{Omitted Details from Section~\ref{sec:uncompressed}}

We explain in this section the reduction mentioned at the beginning of Section~\ref{sec:uncompressed} and start with Gaifman's locality theorem.

\subsection{Gaifman's theorem and a corollary for degree-\boldmath{$d$} bounded structures}
\label{sec-gaifman}

Clearly, for every $r \geq 0$, there is an FO-formula $\delta_{r}(x,y)$ such that for every relational structure $\RS$ and all $a, b \in \RS$ we have: $\RS \models \delta_r(a, b)$ if and only if $\dist_{\RS}(a, b) \leq r$. Moreover, define $\delta_{k,r}(x_1, \ldots, x_k,y) = \bigvee_{i=1}^k \delta_r(x_i,y)$. For an FO-formula $\psi$, a tuple of variables $\bar{z} = (z_1, \ldots, z_k)$ and $r \geq 1$ let $\psi^{\bar{z},r}$ be the \emph{$(\bar{z},r)$-relativization} of $\psi$, i.e., the FO-formula obtained from $\psi$ by restricting all quantifiers to the $r$-sphere around $\bar{z}$ in the graph $\mathcal{G}(\RS)$. Let us define this more formally. For an FO-formula $\psi$, a tuple of variables $\bar{z} = (z_1, \ldots, z_k)$
 and $r \geq 1$ we define the FO-formula $\psi^{\bar{z},r}$ inductively by
\begin{itemize}
\item $R(x_1,\ldots, x_n)^{\bar{z},r} = R(x_1,\ldots, x_n)$ for $R \in \mathcal{R}$ and $(x=y)^{\bar{z},r} = (x=y)$,
\item $(\neg \psi)^{\bar{z},r} = \neg (\psi^{\bar{z},r})$, $(\psi_1 \circ \psi_2)^{\bar{z},r} = \psi_1^{\bar{z},r} \circ \psi_2^{\bar{z},r}$ for $\circ \in \{\wedge, \vee\}$, and
\item $(\exists x : \psi)^{\bar{z},r} = \exists x : (\delta_{k,r}(\bar{z},x) \wedge \psi^{\bar{z},r})$ and $(\forall x : \psi)^{\bar{z},r} = \forall x : (\delta_{k,r}(\bar{z},x) \to \psi^{\bar{z},r})$.
\end{itemize}
Note that the $z_i$ are free variables in $\psi^{\bar{z},r}$. 

We can now state Gaifman's theorem \cite{Gai82}; see also \cite[Theorem~4.22]{Libkin04}.

\begin{theorem} \label{thm-gaifman}
From a given FO-formula $\phi(x_1, \ldots, x_k)$ with $\mathsf{qr}(\phi) = \nu$,
one can compute a logically equivalent Boolean combination of FO-formulas of the following form, where $r \leq 7^{\nu}$, $q \leq k+\nu$ and 
$\bar{x} = (x_1, \ldots, x_k)$:
\begin{enumerate}[(i)]
\item $\psi^{\bar{x},r}$ where only the variables $x_1, \ldots, x_k$ are allowed to occur freely in $\psi$,
\item $\exists z_1 \cdots \exists z_q : \bigwedge_{1 \leq i < j \leq q} \neg\delta_{2r}(z_i,z_j) \wedge \bigwedge_{1 \leq i \leq q} \theta^{z_i,r}$ where 
$\theta$ is a sentence.
\end{enumerate}
\end{theorem}
\medskip
Gaifman's theorem is particularly useful for degree-$d$ bounded structures (for some fixed $d$).
We therefore continue to derive a well-known corollary of Theorem~\ref{thm-gaifman} for 
degree-$d$ bounded structures.

We first observe that for every formula $\psi^{\bar{x},r}$ from (i) in Theorem~\ref{thm-gaifman}, whether $\RS \models \psi^{\bar{x},r}(\bar{a})$ for some $\bar{a} \in \RS^k$ does not depend on the whole structure $\RS$, but is completely determined by the $(k,r)$-neighborhood type of $\bar{a}$. This is due to the fact that all quantifiers in $\psi^{\bar{x},r}$ are restricted to the $r$-sphere around its free variables $\bar{x}$. Consequently, one can replace every formula $\psi^{\bar{x},r}$ from (i) in Theorem~\ref{thm-gaifman} by the finite disjunction $\bigvee_{i \in I} \psi^{\mathcal{B}_i}(\bar{x})$ where the $\mathcal{B}_i \in \mathcal{T}_{k,r}$ for $i \in I$ are exactly those $(k,r)$-neighborhood types such that $\RS \models \psi^{\bar{x},r}(\bar{a})$ if and only if $\bar{a}$ is a $\mathcal{B}_i$-tuple for some $i \in I$.
By going to conjunctive normal form (and using the fact that a conjunction $\psi^{\mathcal{B}_1}(\bar{x}) \wedge 
\psi^{\mathcal{B}_2}(\bar{x})$ for different $\mathcal{B}_1, \mathcal{B}_2 \in \mathcal{T}_{k,r}$ is always false)
we end up with a disjunction 
\[
\bigvee_{i \in [m]} (\psi^{\mathcal{B}_i}(\bar{x}) \wedge \psi_i ),
\]
where for all $i \in [m]$, $\mathcal{B}_i \in \mathcal{T}_{k,r}$ and 
$\psi_i$ is a Boolean combination of sentences of the form (ii) in Theorem~\ref{thm-gaifman} for some $r \leq 7^{\nu}$.
This directly yields the following corollary:
\begin{corollary} \label{coro-gaifmanAppendix}
From a given $d$ and an FO-formula $\phi(x_1, \ldots, x_k)$ with $\mathsf{qr}(\phi) = \nu$,
one can compute a sequence $(\mathcal{B}_1, \psi_1, \ldots, \mathcal{B}_m, \psi_m)$ for some $m = m(d,|\phi|)$
with the following properties:
\begin{enumerate}[(i)]
\item \label{coro-gaifman-(i)} For every $i \in [m]$ there is an $r \leq 7^{\nu}$ such that 
$\mathcal{B}_i \in \mathcal{T}_{k,r}$ and $\psi_i$ is a Boolean combination of sentences of the form (ii) in Theorem~\ref{thm-gaifman}.
\item  \label{coro-gaifman-(ii)} For every degree-$d$ bounded structure $\RS$ and all $\bar{a} \in \RS^k$ we have:
$\RS \models \phi(\bar{a})$ if and only if there is an $i \in [m]$ such that $\RS \models \psi_i$ and 
$\bar a$ is a $\mathcal{B}_i$-tuple.
\end{enumerate}
\end{corollary}
\medskip
Based on Gaifman's locality theorem, Seese \cite{Seese96} proved the following:

\begin{theorem}[degree-bounded model checking]\label{degreeBoundedMCLemma}
For a given FO-sentence $\phi$ and a degree-$d$ bounded structure $\RS$, one can check
in time $|\RS| \cdot f(d,|\phi|)$, whether $\RS \models \phi$.
\end{theorem}
Recall that we have fixed at the beginning of Section~\ref{sec:uncompressed} 
 a relational signature $\mathcal{R} = \{R_i : i \in I\}$, constants $d \geq 2$ and $\nu$, a degree-$d$ bounded structure $\RS = (U, (R_i)_{1 \leq i \leq l})$ over the signature $\mathcal{R}$, and an FO-formula $\phi(x_1, \ldots, x_k)$ over the signature $\mathcal{R}$ with $\mathsf{qr}(\phi) = \nu$. Moreover, our goal is to enumerate the set $\phi(\RS)$ after a linear time preprocessing in constant delay.
By Corollary~\ref{coro-gaifmanAppendix} we can compute 
a sequence $(\mathcal{B}_1, \psi_1, \ldots, \mathcal{B}_m, \psi_m)$ for some $m = m(d,|\phi|)$ with the properties
from point \eqref{coro-gaifman-(i)} in the corollary and then enumerate all $k$-tuples $\bar{a}$ such that
$\RS \models \psi_i$ and $\bar{a}$
is a $\mathcal{B}_i$-tuple for some $i \in [m]$. Using Theorem~\ref{degreeBoundedMCLemma} one can check in time 
$|\RS| \cdot f(d,|\phi|)$ which of the $\psi_i$ are true in $\RS$. We then keep only those 
$\mathcal{B}_i \in \mathcal{T}_{k,r}$ such that $\RS \models \psi_i$. Moreover, w.l.o.g. the remaining 
$\mathcal{B}_i$ are different and for each of them we enumerate all $\mathcal{B}_i$-tuples. This will not create duplicates
(since each $k$-tuple has a unique $(k,r)$-neighborhood type).

\subsection{Reduction to tuples with pairwise different components}
\label{appendix-pairwise-diff}

By the previous discussion, it suffices to enumerate for a fixed 
$\mathcal{B} \in \mathcal{T}_{k,r}$ the set of all $\mathcal{B}$-tuples with delay $f(|\mathcal{B}|)$ after preprocessing time
$|\RS| \cdot f(|\mathcal{B}|)$. Recall that in Section~\ref{sec:uncompressed} we assumed that the sphere center constant $c_i$ 
is interpreted by the element $i \in [k]$ in $\mathcal{B}$.
But this means that all $\mathcal{B}$-tuples have pairwise different components. 

We next show that this is not a real restriction.
More precisely, we claim that w.l.o.g.~we can replace our initial FO-formula $\phi(x_1, \ldots, x_k)$ by 
\begin{equation} \label{eq-phi'} 
\phi' \ = \ \phi \wedge \bigwedge_{i,j \in [k], i \neq j} x_i \neq x_j\,, 
\end{equation}
so that the entries $a_i$ in each enumerated output tuple $(a_1, \ldots, a_k)$ are pairwise different.
For this, the enumeration algorithm runs in an outermost loop through all equivalence relations $\equiv$ on $[k]$. 
For each such equivalence relation $\equiv$ and $i \in [k]$ let $\mu_{\equiv}(i) = \min([i]_{\equiv})$ be the minimal representative 
from the equivalence class $[i]_\equiv = \{ j \in [k] : i \equiv j\}$ 
and let $I_\equiv = \mu_{\equiv}([k])$ be the image of $\mu_{\equiv}$.
Then we define the FO-formula $\phi_{\equiv}$ by replacing every free occurrence of a variable 
$x_i$ in $\phi$ by $x_{\mu_{\equiv}(i)}$ and start the enumeration of the output tuples for the formula
\[
\phi_\equiv' \ = \ \phi_{\equiv} \wedge \bigwedge_{i,j \in I_\equiv, \, i \neq j} x_i \neq x_j.
\]
From every enumerated $|I_\equiv|$-tuple we obtain a tuple in $\phi(\RS)$ by duplicating entries suitably (according to the equivalence relation $\equiv$).
This shows that it suffices to consider a formula of the form $\phi'$ in \eqref{eq-phi'}.

We now apply Corollary~\ref{coro-gaifmanAppendix} to the new formula $\phi'$. By the additional disequality constraints $x_i \neq x_j$
in $\phi'$ we can assume that in the resulting $(k,r)$-neighborhood types $\mathcal{B}_1, \ldots, \mathcal{B}_m$ 
the elements $a_i$ ($i \in [k]$) that interpret the sphere center constants $c_i$ are pairwise different. We can
then w.l.o.g. assume that $c_i$ is interpreted by the element $i \in [k]$ (recall that the universe of   
a $(k,r)$-neighborhood type is of the form  $[\ell]$).

\subsection{Proofs of Lemmas~\ref{correctnessLemma1}~and~\ref{correctnessLemma2}}
\label{lemma-2-3}

As in Section~\ref{sec:uncompressed} we fix a $(k,r)$-neighborhood type $\mathcal{B}$ (with the sphere centers $1, \ldots, k$) that splits into the 
connected components $\mathcal{C}^{\mathcal{B}}_1,
\ldots, \mathcal{C}^{\mathcal{B}}_m$. Moreover, for every $\mathcal{B}'$-node $a \in \RS$ we fix
an isomorphism $\pi_a : \mathcal{B}' \to \neighborhoodStr{\RS}{\bigR}{a}$.

\begin{lemma}[Lemma~\ref{correctnessLemma1} restated]
 If $\Lambda = (\mathcal{B}_1, \sigma_1,  \ldots, \mathcal{B}_m, \sigma_m)$ is a consistent factorization of $\mathcal{B}$
 and $\bar{b} \in \RS^m$ is admissible for $\Lambda$ then $\Lambda(\bar{b})$ is a $\mathcal{B}$-tuple.
 \end{lemma}
 
 \begin{proof} 
Let $D_i = \mathcal{C}^{\mathcal{B}}_i \cap [k]$ and let $\bar{b} = (b_1, \ldots, b_m)$. We have $\neighborhoodStr{\RS}{\bigR}{b_i} \simeq \mathcal{B}_i$.
Moreover, $\Lambda(\bar{b}) = t_{b_1,\sigma_1} \sqcup t_{b_2,\sigma_2} \sqcup \cdots
 \sqcup t_{b_m,\sigma_m}$, where $t_{b_i,\sigma_i}(j) = \pi_{b_i}(\sigma_i(j))$ for all $j \in D_i$.
 Since $\neighborhoodStr{\mathcal{B}_i}{r}{\sigma_i} \simeq \mathcal{C}^{\mathcal{B}}_i$ we obtain 
 $\neighborhoodStr{\RS}{r}{t_{b_i,\sigma_i}} \simeq \mathcal{C}^{\mathcal{B}}_i$
 for all $i \in [m]$. Moreover, from \eqref{disjointness} it follows that the  $\neighborhoodStr{\RS}{r}{t_{b_i,\sigma_i}}$ are 
 the connected components of 
 $\neighborhoodStr{\RS}{r}{t_{b_1,\sigma_1} \sqcup \cdots\sqcup t_{b_m,\sigma_m}} = \neighborhoodStr{\RS}{r}{\Lambda(\bar{b})}$.
 Hence, $\Lambda(\bar{b})$ is a $\mathcal{B}$-tuple.
 \end{proof}
 
 \begin{lemma}[Lemma~\ref{correctnessLemma2} restated]
 If $\bar{a} \in \RS^k$ is a $\mathcal{B}$-tuple then there are a unique consistent factorization $\Lambda$ of 
$\mathcal{B}$ and a unique $m$-tuple $\bar{b} \in \RS^m$ that is admissible for $\Lambda$
 and such that $\bar{a}  = \Lambda(\bar{b})$.
\end{lemma}
 
 \begin{proof} 
Let $D_i = \mathcal{C}^{\mathcal{B}}_i \cap [k]$ and let $n_i \in \min(D_i)$. Moreover, let $\bar{a} = (a_1, \ldots, a_k)$
 and define the partial $k$-tuple $t_i : [k] \to \RS$ with domain $D_i$ by $t_i(j) = a_j$ for $j \in D_i$.
Since $\mathcal{B}$ is the $(k,r)$-neighborhood type of $\bar{a}$, we have
$\mathcal{C}^{\mathcal{B}}_i \simeq \neighborhoodStr{\RS}{r}{t_i} \subseteq \neighborhoodStr{\RS}{\bigR}{a_{n_i}}$ and
$\dist_{\RS}(t_i, t_j) > 2r+1$
for all $i, j  \in [m]$ with $i \neq j$. 
 
 We take $\bar{b} = (b_1, \ldots, b_m)$ with $b_i := a_{n_i}$ for $i \in [m]$.
 Moreover, for every $i \in [m]$, let $\mathcal{B}_i$ be the $\bigR$-neighborhood type of $b_i$ in $\RS$  and 
 define the partial $k$-tuple $\sigma_i : [k] \to \mathcal{B}_i$ with domain $D_i$ by $\sigma_i(j) = \pi_{b_i}^{-1}(a_j)$ for $j \in D_i$.
Moreover, define $\Lambda = (\mathcal{B}_1, \sigma_1,  \ldots, \mathcal{B}_m, \sigma_m)$. The definition of $\sigma_i$ implies
 $t_i = t_{b_i,\sigma_i}$ for all $i \in [m]$ and hence $\bar{a} = \Lambda(\bar{b})$.
 Note that $\sigma_i(n_i) = \pi^{-1}_{b_i}(a_{n_i}) = \pi^{-1}_{b_i}(b_i) = 1$ and
 $\neighborhoodStr{\mathcal{B}_i}{r}{\sigma_i} \simeq \neighborhoodStr{\RS}{r}{t_i} \simeq \mathcal{C}^{\mathcal{B}}_i$.
 Hence, $\Lambda$ is a consistent factorization of 
$\mathcal{B}$. Moreover, $(b_1, \ldots, b_m)$ is admissible for $\Lambda$:
$b_i$ is a $\mathcal{B}_i$-node and 
$\dist_{\RS}(t_{b_i,\sigma_i}, t_{b_{j},\sigma_{j}}) =  \dist_{\RS}(t_i,t_j) > 2r+1$
for all $i, j  \in [m]$ with $i \neq j$.

To show uniqueness, assume that $\Lambda' = (\mathcal{B}'_1, \sigma'_1,  \ldots, \mathcal{B}'_m, \sigma'_m)$ is a consistent factorization of 
$\mathcal{B}$ and $\bar{b}' \in \RS^m$ is admissible for $\Lambda'$
 such that $\bar{a}  = \Lambda(\bar{b}')$.
 First, notice that $\bar{b}' = (a_{n_1}, \ldots, a_{n_m}) = \bar{b}$.
 Moreover, we must have $\mathcal{B}_i = \mathcal{B}'_i$ (both are 
 the $(k,r)$-neighborhood type of $a_{n_i}$). Finally, for all $i \in [m]$ and $j \in D_j$ we have 
$\sigma_i(j) = \pi_{b_i}^{-1}(a_j) = \pi_{b'_i}^{-1}(a_j) = \sigma'_i(j)$ for all $j \in D_i$.
\end{proof}

\subsection{Proof of Lemma~\ref{delayBoundLemma}}\label{sec:delayBoundAppendix}

We now formally prove that the enumeration algorithm described in Section~\ref{sec:enumAlgo} has a delay of $f(d, |\phi|)$. 

Recall from Section~\ref{sec:enumAlgo} that we do this under the assumption that we can check whether a stack $s$ is admissible in time $f(d, |\phi|)$. This will be explained in detail for the more general SLP-compressed setting in Section~\ref{sec-enum-compressed}. Hence, in the following, we only count the number of steps in the
DFLR-traversal, i.e., the total number of iterations of the for-loop in Algorithm~\ref{mainEnumAlgoNonCompressed} over all 
recursion levels. In the following, we use all notations introduced in Section~\ref{sec:enumAlgo}.

We first make the following general observation: Let $s = a_1 \cdots a_{\ell}$ with $\ell < m$ be an admissible stack. If, for some element $a_{\ell + 1} \in L_{\ell + 1}$, the stack $s a_{\ell+1}$ is not admissible, then 
\[ \dist_{\RS}(t_{a_i,\sigma_i}, t_{a_{\ell+1},\sigma_{\ell+1}}) \leq 2r+1 \]
for some $i \in [\ell]$.   
This implies $\dist(a_i, a_{\ell + 1}) \leq 2\bigR+2r+1$. The number of elements $a \in \RS$ such that
$\dist(a_i, a) \leq 2\bigR+2r+1$ for some $i \in [\ell]$ is bounded by $\ell \cdot d^{2\bigR+2r+2} \leq k \cdot d^{2\bigR+2r+2}$.
Hence, we obtain the following observation:

\begin{observation} \label{observation1}
For a fixed admissible stack $s$ of length $\ell < m$, the list $L_{\ell+1}$ contains at most $\ell \cdot d^{2\bigR+2r+2} \leq k \cdot d^{2\bigR+2r+2}$
 elements $a$ such that $sa$ is not admissible.
\end{observation}
In order to bound the delay, let us first consider the time that elapses from the call {\sf extend}$(\varepsilon)$  to the first output 
(or the termination of the algorithm). Recall that the lists $L_1, \ldots, L_q$ are short. Hence, there are at most $c := (k \cdot d^{2\bigR+2r+2})^{q+1} \leq k^{k+1} d^{(2\bigR+2r+2)(k+1)}$ 
many stacks $s$ with $|s| \leq q$. Hence, either the algorithm terminates after at most $c$ many steps in the DFLR-traversal or it reaches after at most $c$ steps a stack $s a$ where $|s|=q$ and $a$ is the first element of the list $L_{q+1}$. Then, by Observation~\ref{observation1}, the algorithm reaches after $(m-q) \cdot k \cdot d^{2\bigR+2r+2} \leq k^2 \cdot d^{2\bigR+2r+2}$ further steps an admissible stack of length $m$ (which is then written to the output). 

Now, let us assume that we output an admissible stack $s$ of length $m$. We have to bound the number of steps of the algorithm until the next output is generated (or until the termination of the algorithm if no further output is produced). The algorithm will first reduce the length of the stack $s$ (by doing return-statements in line 4). First assume that the length of the stack does not go below $q+1$.
Then, Observation~\ref{observation1} shows that the next output happens after at most $2 (m-q) \cdot k \cdot d^{2\bigR+2r+2}$ DFLR-traversal steps, i.e., 
$(m-q) \cdot k \cdot d^{2\bigR+2r+2}$ steps in the first phase, where the stack shrinks, followed by $(m-q) \cdot k \cdot d^{2\bigR+2r+2}$ steps where the stack
grows. On the other hand, if the stack length reaches $q$ (this can only happen after at most $(m-q) \cdot k \cdot d^{2\bigR+2r+2}$ steps) then either the algorithm terminates
after $c$ more steps or after $c + (m-q) \cdot k \cdot d^{2\bigR+2r+2}$ more steps the next output occurs.
Hence, the delay can be bounded by $f(d, |\phi|)$.

\section{Omitted details from Section~\ref{S hier}}\label{SLPDefinitionsAppendix}

Consider an SLP $D=(\mathcal R, N, S, P)$ and a nonterminal $A \in N$. In Section~\ref{S hier}, we have explained on an intuitive level how the structure $\eval(A)$ is defined. Let us now define this more formally. 
We start with a few general definitions.

Let $\equiv$ be an equivalence relation on a set $U$. 
Then, for $a\in U$, $[a]_\equiv = \{ b \in U : a \equiv b\}$
denotes the equivalence class containing $a$. With 
$[U]_\equiv$ we denote the set of all equivalence classes.
With $\pi_\equiv : U \to [U]_\equiv$ we denote the function
with $\pi_\equiv(a) = [a]_\equiv$ for all $a\in U$.

For a relational structure $\RS = (U, (R_i)_{i \in I})$ and 
an equivalence relation $\equiv$ on $U$ we define the
quotient $\RS/_{\equiv} \;= ( [U]_{\equiv}, (R_i/_{\equiv})_{i \in I})$, where 
\[ R_i/_{\equiv} \;=  \{ (\pi_\equiv(v_1),\ldots,\pi_\equiv(v_{\alpha_i})) : (v_1,\ldots,v_{\alpha_i}) \in R_i\}.\]
For two structures $\RS_1 = (U_1, (R_{i,1})_{i \in I})$ and 
$\RS_1 = (U_2, (R_{i,2})_{i \in I})$ over the same signature $\mathcal R$ 
and with disjoint universes $U_1$ and $U_2$, respectively, we 
define the disjoint union 
\[\RS_1 \oplus \RS_2 = (U_1 \uplus U_1, (R_{i,1} \uplus
R_{i,2})_{i \in I}).
\]
With these definitions we can now define the  $\rank(A)$-pointed structure $\eval(A)$ for $A \in N$ as follows.
Assume that $E_A = \{ (A_i,\sigma_i) : i \in [n]\}$. 
Recall that this is a multiset, i.e., we may have $(A_i,\sigma_i) = (A_j,\sigma_j)$ for 
 $i \neq j$. 
 Assume now that
 $\eval(A_i) = (\RS_i,\tau_i)$
is already defined for every $i \in [n]$. 
Then $$\eval(A) = ( (\RS_A \oplus \RS_1 \oplus \cdots \oplus \RS_n)/_{\equiv},
\pi_{\equiv} \circ \tau_A),$$ where 
$\equiv$ is the smallest equivalence relation on the universe of 
$\RS_A \oplus \RS_1 \oplus \cdots \oplus \RS_n$, which contains 
$\{ (\sigma_i(j), \tau_i(j)) : i \in [n], j \in [\rank(A_i)] \}$.
Note that $\eval(A) = (\RS_{A}, \tau_{A})$ if $E_A$ is empty. 

Next, let us take a closer look at our general assumption that the signature $\mathcal{R}$ only contains relation symbols of arity at most two
(see the end of Section~\ref{sec-FOL}). We have seen in Appendix~\ref{sec:generalDefinitionsAppendix} how we can transform an FO-formula $\psi$ and a relational structure $\RS$ over a signature $\mathcal{R}$ (with relation symbols of arbitrary arity) into an FO-formula $\psi'$ and a relational structure $\RS'$ over a signature $\mathcal{R}'$ with all relation symbols of arity at most two, such that $\psi(\RS) = \psi'(\RS')$ and, furthermore, $\mathsf{qr}(\psi') =  \mathsf{qr}(\psi)+1$.

For the compressed setting, we must argue that this reduction can be performed in linear time also in the case where $\RS$ is compressed by an SLP, which is rather simple:
From an SLP $D$ (whose signature $\mathcal{R}$
contains relation symbols of arbitrary arity) one can construct a new SLP $D'$ (whose signature contains only 
relation symbols of arity at most two) such that $\eval(D') = \eval(D)'$. For this it suffices to replace every structure 
$\RS_A$ (for $A$ a nonterminal of $D$) by $\RS'_A$ (where $\RS'_A$ is obtained from $\RS_A$ by the construction described in Appendix~\ref{sec:generalDefinitionsAppendix}). Therefore, it is justified also in the compressed setting to consider only SLPs that produce structures with relations of arity at most two.

\section{Omitted details from Section~\ref{sec-enum-compressed}} 

\subsection{The Gaifman locality reduction in the compressed setting}\label{sec:GaifmanReductionCompressed}

In this section, we explain why the reduction described in Appendix~\ref{sec-gaifman} is also applicable in the compressed setting.
As in the uncompressed setting we want to reduce the enumeration of the result set $\phi(\eval(D))$ for a given 
FO-formula $\phi(x_1, \ldots, x_k)$ and an SLP $D$ to the enumeration of all $\mathcal{B}$-tuples in $\eval(D)^k$ for some 
$(k,r)$-neighborhood type $\mathcal{B}$ with $r \leq 7^{\nu}$ (where $\nu = \mathsf{qr}(\phi)$).

In the uncompressed setting, we apply Corollary~\ref{coro-gaifmanAppendix} to the formula $\phi(x_1, \ldots, x_k)$
and obtain the list $(\mathcal{B}_1, \psi_1, \ldots, \mathcal{B}_m, \psi_m)$ with the properties described in 
Corollary~\ref{coro-gaifmanAppendix}. We then 
check for each sentence $\psi_i$, whether it is true in $\RS$ (using 
 Theorem~\ref{degreeBoundedMCLemma}). We keep only those neighborhood types $\mathcal{B}_i$
 such that $\RS \models \psi_i$ holds and then run the enumeration algorithm for those $\mathcal{B}_i$.
 
For the SLP-compressed setting we have no linear time model checking algorithm (for a fixed FO-sentence) in the spirit of  
Theorem~\ref{degreeBoundedMCLemma} available; only an {\sf NL}-algorithm for apex SLPs is known \cite{Lohrey2012}.
With some effort, one can adapt this {\sf NL}-algorithm to  
 to obtain a linear time algorithm, but there is an easier solution that we explain
below.

Assume that we have already an algorithm that enumerates after preprocessing time $|D| \cdot f(d, \phi)$ in
delay $f(d, \phi)$ all $\mathcal{B}$-tuples from $\eval(D)^k$ for a fixed $(k,r)$-neighborhood type $\mathcal{B}$ (for some $r \leq 7^{\nu}$).
We then use this algorithm for 
\begin{enumerate}[(i)]
\item testing whether $\eval(D) \models \psi_i$ for a $\psi_i$ from 
Corollary~\ref{coro-gaifmanAppendix} and, in case this holds,
\item for enumerating all $\mathcal{B}_i$-tuples 
from $\eval(D)^k$. 
\end{enumerate}
Let us show how we one achieve (i).
Consider one of the sentences $\psi_i$  from Corollary~\ref{coro-gaifmanAppendix}.
It is a Boolean combination of sentences of the form
\begin{equation} \label{eq-local-sentence}
\exists z_1 \cdots \exists z_q : \bigwedge_{1 \leq i < j \leq q} \neg\delta_{2r}(z_i,z_j) \wedge \bigwedge_{1 \leq i \leq q} \theta^{z_i,r}
\end{equation}
with $r \leq 7^{\nu}$, $q \leq k+\nu$ and a sentence $\theta$.
So, we have to check whether there exist at least $q$ many disjoint $r$-neighborhoods (of single nodes) for which the FO-property
$\theta^{z,r}$ holds. From $\theta^{z,r}$ one can compute a finite list $\mathcal{B}'_1, \ldots, \mathcal{B}'_s$ (for some 
$s = f(d, |\phi|)$) of $r$-neighborhood types such that $\theta^{z,r}(z)$ is equivalent to the fact that $z$ has one of the $r$-neighborhood types
$\mathcal{B}'_1, \ldots, \mathcal{B}'_s$.
Then, by using the assumed enumeration algorithm, we enumerate the set $L'_1$ of all $\mathcal{B}'_1$-nodes in $\eval(D)$, then the set $L'_2$ of all $\mathcal{B}'_2$-nodes in $\eval(D)$ and so on, until we either have enumerated $q \cdot d^{2r+1}$ many nodes in total, or the enumeration of $L'_s$ terminates. Hence, in total, we enumerate at most $q \cdot d^{2r+1}$ many nodes. The preprocessing for these enumerations takes time $|D| \cdot f(d,\phi)$ and all enumerations take time $f(d,|\phi|)$ in total. If we actually enumerate $q \cdot d^{2r+1}$ many nodes, then, since $\eval(D)$ is degree-$d$ bounded, there must exist $q$ disjoint $r$-neighborhoods in $\eval(D)$ with a type from $\{\mathcal{B}'_1, \ldots, \mathcal{B}'_s\}$. Indeed, if we have a list of $q \cdot d^{2r+1}$  enumerated nodes, then we pick the first element $a_1$ and 
remove from the list all nodes from the sphere $\sphereStr{\eval(D)}{2r}{a_1}$.
Then we pick the first element $a_2$ from the remaining list and remove all nodes from the sphere $\sphereStr{\eval(D)}{2r}{a_2}$.
We proceed like this and pick elements $a_3, a_4, \ldots$ until the list is empty. Obviously, the picked nodes have pairwise distance of at least $2r+1$ and have all an $r$-neighborhood type $\mathcal{B}'_i$ for some $i \in [s]$. Moreover, since we remove in each step at most $d^{2r+1}$ nodes (namely the sphere $\sphereStr{\eval(D)}{2r}{a_i}$), it is guaranteed that we select at least $q$ such nodes (recall that we start with
a list of $q \cdot d^{2r+1}$ nodes). Hence, the sentence \eqref{eq-local-sentence} holds in $\eval(D)$.  

If, on the other hand, we enumerate strictly fewer than $q \cdot d^{2r+1}$ nodes, then we can 
check for all pairwise distinct enumerated nodes $a,b$ whether the distance in $\mathcal{G}(\eval(D))$ between $a$ and $b$ is at most $2r+1$.
It is shown in Appendix~\ref{sec:disjointnessCheck} that this can be done in time $f(d,r)$; see also Remark~\ref{remark disjoint spheres}.

\subsection{Omitted details from Section~\ref{sec:expansions}}\label{sec:algoPrelimAppendix}

\subsubsection{Top-level nodes and subsets} 

In this section we introduce some concepts that are useful for the proof of Lemma~\ref{bijectionLemma} in the next section.

For a nonterminal $A$, we say that a node $a=(p,v)$ of $\eval(D)$ is \emph{$A$-produced}, if the path $p$ contains $A$. The unique $S$-to-$A$ prefix of $p$ is then called the \emph{$A$-origin} of $a$. Intuitively speaking, every $A$-produced node of $\eval(D)$ is produced by an occurrence of the nonterminal $A$ in the derivation tree, and if two such nodes have the same $A$-origin, then they are produced by the same occurrence of $A$. We say that $a$ is \emph{top-level $A$-produced} if $p$ ends in $A$. By the definition of $D$-representations 
(see Section~\ref{sec node reps}), this means that $v \in \RS_A \setminus \ran(\tau_A)$. Moreover, for every node $a$ of $\eval(D)$ there is a unique nonterminal $A$ such that $a$ is a top-level $A$-produced. 

A subset $U$ of $\eval(D)$ is an \emph{$A$-subset}, if all its nodes are $A$-produced with the same $A$-origin, which is then also called the $A$-origin of $U$. 
If additionally there is at least one node of $U$ that is top-level $A$-produced, then we say that $U$ is a \emph{top-level $A$-set}. Note that this is equivalent to the existence of at least one node in $U$ with a $D$-representation $(p, v)$, where
$p$ ends in $A$ and $v$ is an internal node of $\RS_A$.

\begin{lemma} \label{lemma-top-level-unique}
If $U$ is a top-level $A$-subset, then there is no other nonterminal $B$ such that $U$ is a top-level $B$-subset.
\end{lemma}

\begin{proof}
Assume that for some nonterminals $A, B \in N$ with $A \neq B$, $U$ is both a top-level $A$-subset with $A$-origin $p$ and a top-level $B$-subset with $B$-origin $q$. This means that $p$ is a proper prefix of $q$ or the other way around; without loss of generality, we assume that $p = q p'$, where $p'$ is a $B$-to-$A$ path. This means that every node in $U$ has the $A$-origin $q p'$, which implies that there is no top-level $B$-produced node
in $U$; a contradiction to the assumption that $U$ is a top-level $B$-structure.
\end{proof}
Note that even though for every subset $U$ of $\eval(D)$ there is a nonterminal $A$ such that $U$ is an $A$-subset (for example, every subset is an $S$-subset), $U$ is not necessarily a top-level $A$-subset for some nonterminal $A$. If, however, $U$ is a \emph{connected} subset (in the sense that the substructure induced by $U$ is a connected set in the graph $\eval(D)$, then we have the following:

\begin{lemma} \label{lemma-topmost}
Let $D$ be an apex SLP and $U$ be a connected subset of $\eval(D)$. 
Then there is a unique $A \in N$ such that $U$ is a top-level $A$-subset.
\end{lemma}

\begin{proof}
Let $(p_1, v_1), \ldots, (p_n,v_n)$ be the $D$-representations of all the elements of $U$ and let $p$ be the longest common prefix of the initial paths $p_1, \ldots, p_n$. Assume that the initial path $p$ is an $S$-to-$A$ path.
Hence, all nodes of $U$ are $A$-produced with the same $A$-origin $p$, which means that $U$ is an $A$-subset. We next show that it is also a top-level $A$-subset, i.e., we show that $U$ contains a node with $D$-representation $(p, v)$, where 
$v$ is an internal node of $\RS_A$. 

In order to get a contradiction, assume that such a node does not exist in $U$. By the choice of $p$ as a longest common prefix, there 
must exist $i,j \in [n]$ with $i \neq j$ such that $p_i = p \, k_i \,B_i \, q_i$ and
 $p_j = p\, k_j \,B_j \,q_j$ with $k_i \neq k_j$. 
 Let $(B_i, \sigma_i)$ be the $k_i^{\text{th}}$ reference in $E_A$ and 
 $(B_j, \sigma_j)$ be the $k_j^{\text{th}}$ reference in $E_A$.
Since $U$ is connected there is a path from $(p \, k_i \,B_i \,q_i, v_i)$ to $(p \,k_j\, B_j\, q_j, v_j)$ in the undirected graph
$\mathcal{G}(\eval(D))$. This path must contain a node $(p, \sigma_i(\ell))$ from $\eval(A)$ (as well as a node $(p, \sigma_j(\ell))$ from $\eval(A)$). Since $D$ is an apex SLP, $\sigma_i(\ell)$ is internal in $\RS_A$, which gives the desired contradiction.

This means that $U$ is a top-level $A$-subset and by Lemma~\ref{lemma-top-level-unique}, $U$ cannot be a top-level $B$-subset for some $B \in N$ with $A \neq B$.
\end{proof}
Note that if $U$ is the universe of a valid substructure (see Section~\ref{sec:expansions}) of $\expan_{\zeta}(A)$, 
then for every $S$-to-$A$ path $p$, the set of $\eta_p(U)$ is a top-level $A$-subset.

Also note that if $U$ is a subset of $\eval(A)$, $p$ is an $S$-to-$A$ path in $\eval(D)$, and $\eta_p(U)$ is a top-level $A$-subset, then
the substructure of $\eval(A)$ induced by $U$ is isomorphic to the substructure of $\eval(D)$ induced by $\eta_p(U)$. The reason is
that $U$ contains none of the contact nodes of $\eval(A)$. In particular, if $\mathcal{A}$ is a valid substructure of the expansion
$\expan_{\zeta}(A)$ then, by definition, it is an induced substructure, and we obtain $\mathcal{A} \simeq \eta_p(\mathcal{A})$.

\subsubsection{Proof of Lemma~\ref{bijectionLemma}}

Let us recall that 
\begin{itemize}
\item $\mathsf{S}^{\mathcal{B}}_1 = \{ (p,a) : \exists A \in N : \text{ $p$ is an $S$-to-$A$ path in $\DAG(D)$ and $a$ is a valid $\mathcal{B}$-node in $\expan(A)$} \}$,
\item $\mathsf{S}^{\mathcal{B}}_2 = \{ b \in \eval(D) : b \text{ is a $\mathcal{B}$-node} \}$,
\end{itemize}
and, for all $(p,a) \in \mathsf{S}^{\mathcal{B}}_1$,  $h(p, a) = \eta_p(a) = (pq, v)$ if  $p$ is an $S$-to-$A$ path in $\DAG(D)$ and $(q,v)$ is the $A$-representation of the valid $\mathcal{B}$-node  $a \in \expan(A)$. 

We prove that $h$ is a bijection from $\mathsf{S}^{\mathcal{B}}_1$ to $\mathsf{S}^{\mathcal{B}}_2$
 (i.\,e., Lemma~\ref{bijectionLemma}) by proving the following three lemmas.

\begin{lemma} \label{lemma-exp1} 
For every $(p,a) \in \mathsf{S}^{\mathcal{B}}_1$ we have $h(p,a) \in \mathsf{S}^{\mathcal{B}}_2$.
\end{lemma}

\begin{proof}
Assume that $p$ is an $S$-to-$A$ path in $\DAG(D)$. Hence, $a$ is a node of $\expan(A)$ such that 
$\neighborhoodStr{\expan(A)}{\bigR}{a} \simeq \mathcal{B}$ and 
$\neighborhoodStr{\expan(A)}{\bigR}{a}$ is a valid substructure of $\expan(A)$.
To prove the lemma, we show that
\begin{equation}
 \label{eq-lemma-exp1}
\neighborhoodStr{\eval(D)}{\bigR}{\eta_{p}(a)} = \eta_{p}(\neighborhoodStr{\expan(A)}{\bigR}{a}).
\end{equation}
Since $\eta_{p}(\neighborhoodStr{\expan(A)}{\bigR}{a}) \simeq \neighborhoodStr{\expan(A)}{\bigR}{a} \simeq \mathcal{B}$
this yields $\neighborhoodStr{\eval(D)}{\bigR}{h(p,a)} = \neighborhoodStr{\eval(D)}{\bigR}{\eta_{p}(a)} \simeq \mathcal{B}$.

Let $(q,v)$ be the $A$-representation of $a$. Thus the path $q$ starts with $A$.
The node $a = (q,v)$ is not of the form $(A, \tau_A(i))$ for some $i \in [\rank(A)]$ (otherwise we would have $a \in \mathsf{Bd}_{A,2\bigR+1}$
and $\neighborhoodStr{\expan(A)}{\bigR}{a}$ would not be valid). Therefore, we have
 $\eta_{p}(a) = (pq,v) \in \eval(D)$. 
Since both substructures in \eqref{eq-lemma-exp1} are induced substructures of $\eval(D)$, it suffices to show that $\eta_{p}(\sphereStr{\expan(A)}{\bigR}{a}) = \sphereStr{\eval(D)}{\bigR}{\eta_{p}(a)}$.
The inclusion $\eta_{p}(\sphereStr{\expan(A)}{\bigR}{a}) \subseteq \sphereStr{\eval(D)}{\bigR}{\eta_{p}(a)}$ holds since 
every $\mathcal{G}(\expan(A))$-path (i.e., a path in the undirected graph $\mathcal{G}(\expan(A))$)
is mapped by $\eta_p$ to a corresponding $\mathcal{G}(\eval(D))$-path of the same length.

Assume now that there is a node $b' \in \sphereStr{\eval(D)}{\bigR}{\eta_{p}(a)} \setminus \eta_{p}(\sphereStr{\expan(A)}{\bigR}{a})$.
We will deduce a contradiction.
There is a $\mathcal{G}(\eval(D))$-path of length at most $\rho$ from $\eta_{p}(a)$ to $b'$. W.l.o.g.~we
can assume that all nodes
along this path except for the final node $b'$ belong to the set $\eta_{p}(\sphereStr{\expan(A)}{\bigR}{a})$. In particular, 
there is a node $\eta_{p}(b) \in \eval(D)$ with $b  \in \sphereStr{\expan(A)}{\bigR}{a}$ such that  there is a $\mathcal{G}(\eval(D))$-path
$\Pi_{\eta_{p}(a),\eta_{p}(b)}$ of length at most $\bigR-1$ from $\eta_{p}(a)$ to $\eta_{p}(b)$ and all nodes along this path belong to the set
$\eta_{p}(\sphereStr{\expan(A)}{\bigR}{a})$. Moreover, there is an edge in $\mathcal{G}(\eval(D))$
between $\eta_{p}(b)$ and $b'$. 
The $\mathcal{G}(\eval(D))$-path $\Pi_{\eta_{p}(a),\eta_{p}(b)}$ yields a corresponding $\mathcal{G}(\expan(A))$-path $\Pi_{a,b}$
of the same length (hence, at most $\bigR-1$) from $a$ to $b$.\footnote{We use here the following general fact: If we have 
an edge $(a,b)$ in the graph $\mathcal{G}(\RS)$ for a structure $\RS$ such that $a$ and $b$ belong to a subset $V$ of $\RS$,
then $(a,b)$ is also an edge in the graph $\mathcal{G}(\mathcal{V})$, where $\mathcal{V}$ is the substructure induced by $V$. This statement
is true since we restrict to relational structures where all relations have arity at most two, but it is
wrong if we allow relations of arity 3 or larger and  $\mathcal{G}(\RS)$ would be the Gaifman graph of a structure $\RS$.
Take for instance the structure $\RS$ with three elements $a,b,c$ and the ternary relation $\{ (a,b,c) \}$.
In $\mathcal{G}(\RS)$ there is an edge between $a$ and $b$, but there is no such edge in $\mathcal{G}(\mathcal{V})$, where $\mathcal{V}$ is the substructure of $\RS$ induced by $\{a,b\}$.}
We therefore have 
$\dist_{\expan(A)}(a,b) \leq \bigR-1$.

We claim that $b' \in \eta_{p}(\expan(A))$: We have $b \in \neighborhoodStr{\expan(A)}{\bigR}{a}$ and $\neighborhoodStr{\expan(A)}{\bigR}{a}$ is a valid substructure of $\expan(A)$.
Moreover, if $\mathcal{C}$ is any valid substructure of $\expan(A)$ then 
all neighbors of $\eta_{p}(\mathcal{C})$ in $\mathcal{G}(\eval(D))$ belong to $\eta_{p}(\expan(A))$ (this is the crucial property
of valid substructures). 
Since $b'$ is a neighbor of $\eta_{p}(b)$ in $\mathcal{G}(\eval(D))$, we obtain $b' \in \eta_{p}(\expan(A))$.
Let $b' = \eta_{p}(c)$ with $c \in \expan(A)$. Then $(b,c)$ is an edge in $\mathcal{G}(\expan(A))$ and hence
$c \in \sphereStr{\expan(A)}{\bigR}{a}$, i.e., $b'  = \eta_p(c) \in \eta_{p}(\sphereStr{\expan(A)}{\bigR}{a})$, which is a contradiction.
\end{proof}

\begin{lemma} \label{lemma-exp2} 
$h : \mathsf{S}^{\mathcal{B}}_1  \to \mathsf{S}^{\mathcal{B}}_2$ is surjective.
\end{lemma}

\begin{proof}
Let $b = (p',v)$ be a $\mathcal{B}$-node of $\eval(D)$. 
Let $A$ be the unique nonterminal $A$ such that $\sphereStr{\eval(D)}{\bigR}{b}$ is a top-level $A$-subset (see Lemma~\ref{lemma-topmost}), and let $p$ be the $A$-origin of $\sphereStr{\eval(D)}{\bigR}{b}$.
We obtain a factorization $p' = pq$ where $p$ ends in $A$ and $q$ starts in $A$. 
Let $a = (q,v) \in \eval(A)$, which is an $A$-representation. 
It remains to show that $(p,a) \in \mathsf{S}^{\mathcal{B}}_1$.

Since $\sphereStr{\eval(D)}{\bigR}{b}$ is a top-level $A$-subset, we have
$\eta_{p}(\neighborhoodStr{\eval(A)}{\bigR}{a}) = \neighborhoodStr{\eval(D)}{\bigR}{b} \simeq \mathcal{B}$ and
hence $\neighborhoodStr{\eval(A)}{\bigR}{a} \simeq \mathcal{B}$. 
We next show that $\neighborhoodStr{\eval(A)}{\bigR}{a}$ is a valid substructure of $\expan(A)$,
which implies $\neighborhoodStr{\expan(A)}{\bigR}{a} = \neighborhoodStr{\eval(A)}{\bigR}{a} \simeq \mathcal{B}$. 
This shows that $a$ is a valid $\mathcal{B}$-node in  $\expan(A)$.

In order to show that $\neighborhoodStr{\eval(A)}{\bigR}{a}$ is a valid substructure of $\expan(A)$, we first observe that since $\sphereStr{\eval(D)}{\bigR}{b}$ is a top-level $A$-subset, we have 
 \[ \sphereStr{\eval(A)}{\bigR}{a} \cap \{ (A, v) : v \in \ran(\tau_A) \} = \emptyset\] and 
there is a node $a_0 \in \sphereStr{\eval(A)}{\bigR}{a} \cap\mathsf{In}_A$.
 We have $\dist_{\eval(A)}(a_0,a) \leq \bigR$. Hence, for every node $c \in \sphereStr{\eval(A)}{\bigR}{a}$ we have
$\dist_{\eval(A)}(a_0,c) \leq 2\bigR$. Therefore, $\neighborhoodStr{\eval(A)}{\bigR}{a}$ is a substructure of $\expan(A)$ and
 $\neighborhoodStr{\eval(A)}{\bigR}{a} \cap \{ c \in \eval(A) : \dist_{\eval(A)}(\mathsf{In}_A, c) = 2\bigR+1 \} = \emptyset$.
\end{proof}

\begin{lemma} \label{lemma-exp3} 
$h : \mathsf{S}^{\mathcal{B}}_1  \to \mathsf{S}^{\mathcal{B}}_2$ is injective.
\end{lemma}

\begin{proof}
Let $(p_1, a_1), (p_2, a_2) \in \mathsf{S}^{\mathcal{B}}_1$ such that $\eta_{p_1}(a_1) = h(p_1,a_1) = h(p_2, a_2) = \eta_{p_2}(a_2)$.
We show that $(p_1, a_1) = (p_2, a_2)$. 
Assume that $p_i$ is an $S$-to-$A_i$ path in $\DAG(D)$ ($i \in \{1,2\}$).
Let $(q_i, v_i)$ be the $A_i$-representation of $a_i$. Hence, we have $(p_1 q_1, v_1) = (p_2 q_2, v_2)$, i.e., $v_1 = v_2$
and $p_1 q_1 = p_2 q_2$. In order to show that $p_1 = p_2$ and $q_1 = q_2$ it suffices to show that $A_1 = A_2$ (then
$p_1$ and $p_2$ end in the same nonterminal and therefore must be equal). 

Recall that $\eta_{p_i}(\sphereStr{\expan(A_i)}{\rho}{a_i})$ is a top-level $A_i$-subset of $\eval(D)$.
By \eqref{eq-lemma-exp1} we have
\[ \eta_{p_1}(\sphereStr{\expan(A_1)}{\rho}{a_1}) = \sphereStr{\eval(D)}{\rho}{\eta_{p_1}(a_1)}=
 \sphereStr{\eval(D)}{\rho}{\eta_{p_2}(a_2)}=
\eta_{p_2}(\sphereStr{\expan(A_2)}{\rho}{a_2}).
\]
 Lemma~\ref{lemma-topmost}
yields $A_1 = A_2$. 
\end{proof}

\subsection{Omitted details from Section~\ref{sec:generalStructure}}

\subsubsection{Computing all  \boldmath{$(2\bigR+1)$}-expansions}
\label{sec compute expansion}

The general idea to compute $\expan(A)$ is to start a breadth-first search (BFS) in all nodes from $\mathsf{In}_A$
 and thereby explore all nodes in $\eval(A)$ with distance of at most $2\bigR+1$ from any node in $\mathsf{In}_A$. Obviously, this BFS has to be able to jump back and forth from one structure $\RS_B$ into another structure $\RS_C$ according to the references in the productions of $D$, which can be realized as follows.

The BFS visits triples $(p, u, \ell)$, where $(p, u)$ is a node of $\eval(A)$, i.e., $p$ is an $A$-to-$B$ path 
 in $\DAG(D)$ for some  $B \in N$ (this includes the case $B = A$, for which $p = A$) and $u \in \RS_B$, where in addition $u \notin \ran(\tau_B)$ in case $B \neq A$, and $\ell$ with $0 \leq \ell \leq 2\bigR+1$ indicates that $(p,u)$ has distance $\ell$ from  $\mathsf{In}_A$.
Initially, we visit all $(A, u, 0)$, where $u$ is an internal node of $\RS_A$. 
From the current triple 
$(p, u, \ell)$, where $p$ ends in $B \in N$ and $\ell \leq 2\bigR$, we can visit unvisited triples of the following form:
\begin{itemize}
\item Stay1: $(p, v, \ell + 1)$ if  $v \in \RS_B$ is adjacent to $u$ in the graph $\mathcal{G}(\RS_B)$, 
where in addition $v \notin \ran(\tau_B)$ in case $B \neq A$  (we stay in the structure $\RS_B$ and move via an edge from $\RS_B$),
\item Stay2: $(p, v, \ell + 1)$ if  $v \in \RS_B$ and there is a reference $(C, \sigma) \in E_B$ with
$u = \sigma(i)$, $v = \sigma(j)$ (for $i,j \in [\mathsf{rank}(C)]$) and in $\RS_C$ there is an edge in 
$\mathcal{G}(\RS_C)$ from $\tau_C(i)$ to $\tau_C(j)$ (we stay in the structure $\RS_B$ but move via an edge that is produced from a reference),
\item Move down: $(p \, j \, C, v, \ell + 1)$ if the $j^{\text{th}}$ reference in $E_B$ is $(C, \sigma)$, there is an 
$i \in [\rank(C)]$ such that $\sigma(i) = u$ and $v \in \RS_C \setminus \ran(\tau_C)$ is adjacent to $\tau_C(i)$ in $\mathcal{G}(\RS_C)$
(we move into a structure $\RS_C$ that is produced from a reference).
\item Move up: $(p' C, \sigma(i), \ell + 1)$ if $p = p' C j B$, the $j^{\text{th}}$ reference in $E_C$ is $(B,\sigma)$, and in $\mathcal{G}(\RS_B)$,
$u$ is adjacent to $\tau_B(i)$ for $i\in [\rank(B)]$; note that $\sigma(i) \notin \ran(\tau_C)$ by the apex condition
(via a contact node of $\RS_B$ we move up to the structure from which the current copy of $\RS_B$ was produced).
 \end{itemize}
This BFS visits exactly those triples $(p, u, \ell)$ such that $(p, u)$ is a node from $\eval(A)$
with distance $\ell \leq 2\bigR+1$ from $\mathsf{In}_A$. Consequently, we visit exactly the nodes of $\expan(A)$. The nodes in $\mathsf{Bd}_{A,2\bigR+1}$ can be easily detected in the BFS.
Note also that for every computed triple $(p, v, \ell)$ the path $p$ has length at most $\ell \leq 2\bigR+1$.

This construction requires time $\bigo(|\expan(A)|) \leq |\RS_A| \cdot f(d, |\phi|)$ 
(the size of a $(2\bigR+1)$-sphere around a tuple of length at most $|\RS_A|$ in a degree-$d$ bounded structure).
Summing over all $A \in N$ shows that all $(2\bigR+1)$-expansions can be computed in time $|D| \cdot f(d, |\phi|)$.

\subsubsection{Computing the \boldmath{$\bigR$}-neighborhood types and useful nonterminals}\label{sec:computeneighborhoodTypes}

For every nonterminal $A$, every node $a \in \expan(A)$, and every $\bigR$-neighborhood type $\mathcal{B}_i$ from 
\eqref{eq-consistent-factorization} we check whether $a$ is a valid $\mathcal{B}_i$-node. For this, we have to compute
$\neighborhoodStr{\expan(A)}{\bigR}{a}$, which can be done in time $f(d,|\phi|)$.
If this is the case, we store $a$ in a list $\mathcal{L}_{i, A}$, and we also 
compute and store an isomorphism $\pi_a : \mathcal{B} \to \neighborhoodStr{\expan(A)}{\bigR}{a}$ that satisfies $\pi_a(1) = a$. 

Since the size of every expansion $\expan(A)$ can be bounded by $|\RS_A| \cdot f(d,|\phi|)$, the total time needed for the above
computations is bounded by $|D| \cdot f(d,|\phi|)$.

Recall that a nonterminal $A$ is $\mathcal{B}_i$-useful if and only if there is a valid $\mathcal{B}_i$-node in $\expan(A)$.
Hence, $A$ is $\mathcal{B}_i$-useful if and only if the list $\mathcal{L}_{i, A}$ is non-empty. Hence, we have also
computed the set of $\mathcal{B}_i$-useful nonterminals.

\subsubsection{Preprocessing for path enumeration}\label{sec:pathEnumPreprocessing}

The following preprocessing is needed for our path enumeration algorithm in $\DAG(D)= (N, \gamma, S)$.
Recall that we denote with $\mathcal{P}_i$ the set of all initial paths in $\DAG(D)$ that end in a $\mathcal{B}_i$-useful nonterminal
(see Section~\ref{sec:generalStructure}),
where $\mathcal{B}_i$ is from the factorization \eqref{eq-consistent-factorization}.

We first extend $\DAG(D)$ to the new dag $\DAG(D)' = (N \uplus N', \gamma', S)$, where $N' = \{ A' : A \in N \}$ is a copy of $N$,
$\gamma'(A) = A' \gamma(A)$  and $\gamma'(A') = \varepsilon$ for every $A \in N$. In other words, we add
for every nonterminal $A$ a new node $A'$ (a leaf in $\DAG(D)'$) that becomes the first child of $A$.
Recall that if $\gamma(A) = B_1 \cdots B_n$ then we write $(A,1,B_1), \ldots, (A,n,B_n)$ for the outgoing 
edges of $A$ in $\DAG(D)$. It is convenient to keep these triples also in $\DAG(D)'$ and to write $(A,0,A')$
for the new edge from $A$ to $A'$.
Clearly, $\DAG(D)'$ can be constructed in time $\bigo(|D|)$ from $D$.

There is a one-to-one correspondence between initial paths in $\DAG(D)$ and initial-to-leaf paths in $\DAG(D)'$.
Formally, if $p$ is an initial-to-$A$ path in $\DAG(D)$, then $p' := p \, 0 \, A'$ is an initial-to-leaf path in $\DAG(D)'$.
Moreover, every initial-to-leaf path in $\DAG(D)'$ is of the form $p \, 0 \, A'$ for a unique initial-to-$A$ path $p$ in $\DAG(D)$.
This means that in the following, we can talk about initial-to-leaf paths of $\DAG(D)'$ instead of initial paths of $\DAG(D)$. Note that in constant time we can always obtain the initial path from $\DAG(D)$ that corresponds to some initial-to-leaf path of $\DAG(D)'$.

Recall the lexicographical ordering of paths with the same starting node from Section~\ref{sec-dag}.
For a path $p$ in $\DAG(D)$ that starts in $A$, we define $\lex_A(p)$ as the position of $p$ in the lexicographically ordered list of all 
paths in $\DAG(D)$ that start in $A$, where $\lex_A(A)=0$ (note that $A$ is the lexicographically smallest path that starts in $A$).
We write $\lex(p)$ for $\lex_S(p)$. Note that for an $S$-to-$A$ path $p$ and an $A$-to-$B$ path $q$ we have
$\lex(pq) = \lex(p)+\lex_A(q)$.

For an initial-to-leaf path $p'$ in $\DAG(D)$ we write $\lex'(p')$ for the position 
of $p'$ in the lexicographically ordered list of all initial-to-leaf paths in $\DAG(D)'$, where we start again with position $0$.
For every initial-to-$A$ path $p$ in $\DAG(D)$ we then have  $\lex(p) = \lex'(p0A')$.
Note that $\lex'$ always refers to the dag $\DAG(D)'$, whereas
$\lex$ always refers to $\DAG(D)$.

Next, we add edge-weights to $\DAG(D)'$ as follows in order to compute $\lex'(p')$ for an initial-to-leaf path $p'$ in $\DAG(D)'$.
 For every node $A$ of $\DAG(D)'$, let $\numberPaths(A)$ be the number of different $A$-to-leaf paths in $\DAG(D)'$. These values can be easily computed in time $\bigo(|\DAG(D)'|) \leq \bigo(|D|)$. Indeed, for every leaf $A'$, we set $\numberPaths(A) = 1$. For an inner node $A$ with $\gamma'(A) = A_0 A_1 \cdots A_m$ (where $A_0 = A'$), we set $\numberPaths(A) = \sum^m_{i = 0} \numberPaths(A_i)$. 

Next, we compute the number $\edgeWeight(e)$ for every edge $e$ of $\DAG(D)'$. Let $A$ be an arbitrary inner node with 
$\gamma'(A) = A_0 A_1 \cdots A_m$ (where $A_0 = A'$). Then, for every $0 \le i \le m$, we set $\edgeWeight(A, i, A_i) = \sum^{i-1}_{j = 0} \numberPaths(A_i)$. Obviously, these numbers $\edgeWeight(e)$ for every edge $e$ can be also computed in time $\bigo(|\DAG(D)'|) \leq \bigo(|D|)$. 

By definition, $\numberPaths(A)$ is the total number of $A$-to-leaf paths in $\DAG(D)'$. Thus, for every $0 \le i \le m$, there are exactly $\edgeWeight(A, i, A_i)$ many $A$-to-leaf paths that are lexicographically smaller than the smallest $A$-to-leaf path that starts with the edge $(A, i,A_i)$.

We extend the function $\edgeWeight$ to paths in $\DAG(D)'$ in the obvious way: for a path $p' = A_0 j_1 A_1 \cdots A_{n-1} j_n A_n$, we set $\edgeWeight(p') = \sum^n_{i = 1} \edgeWeight(A_{i-1}, j_i, A_i)$. It can be verified by induction that for every 
initial-to-leaf path $p'$ in $\DAG(D)'$ we have $\edgeWeight(p') = \lex'(p')$. Hence, for $p' = p0A'$ we obtain
\begin{equation} \label{eq-lex-weight}
\edgeWeight(p0A')  = \lex'(p0A') = \lex(p).
\end{equation}
Next, for every $\bigR$-neighborhood type $\mathcal{B}_i$ from the factorization \eqref{eq-consistent-factorization}, we create a dag $\DAG_i$ from $\DAG(D)'$ as follows. 
Initially, let $\DAG_i$ be just a copy of $\DAG(D)'$. In a first step, we delete all nodes in $\DAG_i$ from which no leaf $A'$ (with $A \in N$) is reachable such that $A$ is $\mathcal{B}_i$-useful in the SLP $D$. We can assume that 
 the initial nonterminal $S$ is not removed here, otherwise 
$D$ contains no $\mathcal{B}_i$-useful nonterminals and the factorization  \eqref{eq-consistent-factorization} will not lead
to any output tuples. We normalize the outgoing edges $(A,i,B)$ for every node $A$ such that the middle indices $i$ form an
interval $[n]$ for some $n$.

The paths from $\mathcal{P}_i$ correspond exactly to the initial-to-leaf paths in $\DAG_i$.
Note that we only delete nodes and edges, but do not change edge-weights when construct $\DAG_i$ from $\DAG(D)'$.
This means that if $p$
is an initial-to-leaf path in $\DAG_i$ then $p$ has the same weight in $\DAG(D)'$ and $\DAG_i$.

For an arbitrary DAG $G = (V,\gamma,\iota)$ and a node $v \in V$ we define
the \emph{outdegree} of $v$ as $|\gamma(v)|$.

A \emph{maximal non-branching} path in $\DAG_i$ is 
a path $p = A \, j \, B_1 \, 1 \, B_2 \, 1 \cdots B_{n-1} \, 1 \, B_n$ where 
$n \geq 1$, $A$ has outdegree at least $2$, every $B_i$ for $i \in [n-1]$ has outdegree $1$ (so $B_{i+1}$ is the unique child of $B_i$), and $B_n$ has outdegree $0$ or at least $2$. 
For technical reasons that shall become clear later on, we want to contract each such maximal non-branching path $p$ into a single edge $(A, j, B_n)$ with $\edgeWeight(A, j, B_n) = \edgeWeight(p)$. This can be done as follows.

We first determine the set $V_1$ of nodes of outdegree one. Then, for every node $A \in V_1$, we compute the unique pair $(\omega(A), g(A))$, where $\omega(A)$ is the unique node with $\omega(A) \notin V_1$ that is reached from $A$ by the unique path $p_A$ consisting of edges $(A_1, 1, A_2)$ with $A_1 \in V_1$, and $g(A) = \edgeWeight(p_A)$. This can be done bottom-up in time $\bigo(|\DAG_i|)$ as follows: For every edge $(A, 1, B)$ with $A \in V_1$ we set $\omega(A) = B$ and $g(A) = \edgeWeight(A, 1, B)$ if $B \notin V_1$ (this includes the case where $B$ is a leaf), and we set $\omega(A) = \omega(B)$ and $g(A) = \edgeWeight(A, 1, B) + g(B)$ if $B \in V_1$ has already been processed. 
We then replace every edge $(B, i, A)$ with $A \in V_1$ by the edge $(B, i, \omega(A))$ with $\edgeWeight(B, i, \omega(A)) = \edgeWeight(B, i, A) + g(A)$. After this step, there is no edge that ends in a node of outdegree one. In particular, if a node has outdegree one, 
then it has no incoming edges. We then remove all nodes of outdegree one and their outgoing edges from $\DAG_i$, except the initial node $S$ (which might have outdegree $1$ as well). Now all nodes except possibly $S$ have outdegree $0$ or at least $2$. If the initial node $S$ has outdegree at least $2$, we are done. If $S$ has outdegree $1$, then, by our construction, it has an edge $(S, 1, A)$ and $A$ has outdegree $0$ or at least $2$. In this case, we delete $S$ and make $A$ the new initial node,  
and we store the information that every initial-to-leaf path of $\DAG_i$ has to be interpreted as the initial-to-leaf path obtained by prepending the edge $(S, 1, A)$. 

Recall that the paths from $\mathcal{P}_i$ in $\DAG(D)$ are in a one-to-one
correspondence with the initial-to-leaf paths of $\DAG_i$. The only difference is that we added an edge $(A,i,A')$ at the end and that
the maximal non-branching subpaths of the initial paths of $\DAG(D)'$ are contracted into single edges in $\DAG_i$. Moreover, the weight of an initial-to-leaf path in $\DAG_i$ is the same as the 
weight of the corresponding path from $\mathcal{P}_i$, which is its $\lex$-value.

Every $\DAG_i$ can be computed in time $\bigo(|\DAG(D)'|) \leq \bigo(|D|)$ and we have to construct $m \leq k \leq |\phi|$ many of them.
Hence, the total running time is bounded by $|D| \cdot f(d,|\phi|)$. 

\subsubsection{Enumerating paths from  \boldmath{$\mathcal{P}_i$ in \boldmath{$\DAG(D)$}}}\label{sec:MainPathEnumSection}

The general idea is that we enumerate with constant delay the initial-to-leaf paths in  $\DAG_i$ in lexicographic order. As explained in Section~\ref{sec:pathEnumPreprocessing}, these paths correspond to the paths from $\mathcal{P}_i$ in
$\DAG(D)$.
In order to be able to go in constant time from an initial-to-leaf path of $\DAG_i$ to the lexicographical next initial-to-leaf
path, we use an idea from \cite{LohreyMR18} (which in turn is based on \cite{GasieniecKPS05}). In \cite{LohreyMR18} only binary dags are considered, i.e., dags where every non-leaf
node has a left and and a right outgoing edge. Then the idea is to represent an initial-to-leaf path in a compact form where every maximal
 subpath $p$ (going from $u$ to $v$)
that consist only of left (right, respectively) edges is contracted into a single triple $(u, \ell, v)$ ($(u, r, v)$, respectively). 

In our dags $\DAG_i$, every non-leaf node can have an arbitrary outdegree larger than one.
We therefore have to adapt the approach from \cite{LohreyMR18} to handle also nodes with more than two outgoing edges. The idea will be to contract only
maximal subpaths consisting of leftmost (rightmost, respectively) edges into single triples.

We describe the procedure for a general dag $G = (V, \gamma, \initial)$ where, for every $v \in V$, we have $|\gamma(v)| = 0$ or $|\gamma(v)| \geq 2$, and every edge $(u, j, v)$ has an edge-weight $\edgeWeight(u, j, v) \in \mathbb{N}$ (note that the dags $\DAG_i$ have this property). This also means that paths of $G$ have weights. 
We call every edge $(u, 1, v)$ a \emph{min-edge}, every edge $(u, |\gamma(u)|, v)$ a \emph{max-edge}, and edges that are neither min- nor max-edges are called \emph{middle-edges}. This means that every node is either a leaf with no outgoing edges, or it is an inner node with one min- and one max-edge and a (possibly zero) number of middle edges.

\subparagraph{Min-max-contracted representations.} 
In this section, it will be convenient to write paths in $G$ as sequences of edge triples $(u,i,v)$ instead of words 
from $(V \mathbb{N})^* V$. Thus the path $v_0 j_1 v_1 j_2 v_2 \cdots j_n v_n$ will be written 
 $(v_0, j_1, v_1) (v_1, j_2, v_2) \cdots (v_{n-1}, j_n, v_n)$. With this notation, the concatenation of paths corresponds to the concatenation of sequences of edges.

Let $p = (v_0, j_1, v_1) \cdots (v_{n-1}, j_n, v_n)$ be a path in $G$. The \emph{min-max-contracted representation} of $p$ is obtained by contracting every maximal sequence of min-edges (max-edges, respectively) into a single triple.
Here, a sequence of min-edges (max-edges, respectively) is \emph{maximal} if it cannot be extended to the left or right by min-edges (max-edges, respectively). More formally, the min-max-contracted representation of $p$ is obtained by replacing each non-empty sequence $(v_i, 1, v_{i+1}) (v_{i+1}, 1, v_{i+2}) \cdots (v_{i'-1}, 1, v_{i'})$ of min-edges that is neither preceded nor followed by a min-edge with the single triple $(v_i, \min, v_{i'})$, and, analogously, replacing each non-empty sequence $(v_i, |\gamma(v_i)|, v_{i+1}) \cdots (v_{i'-1}, |\gamma(v_{i'-1})|, v_{i'})$ of max-edges that is neither preceded nor followed by a max-edge with the single triple $(v_i, \max, v_{i'})$. Intuitively speaking, a triple $(v_i, \min, v_{i'})$ ($(v_i, \max, v_{i'})$, respectively) means that we go from $v_i$ to $v_{i'}$ by only taking the first (last, respectively) outgoing edge for every node. For example, the path 
\begin{equation*}
(v_0, 1, v_1) (v_1, 1, v_2) (v_2, 7, v_3) (v_3, 4, v_4) (v_4, 1, v_5) (v_5, |\gamma(v_5)|, v_6) (v_6, |\gamma(v_6)|, v_7) (v_7, 5, v_8)
\end{equation*}
(where $|\gamma(v_2)| >7$, $|\gamma(v_3)|>4$, and $|\gamma(v_7)|>5$)
has the min-max-contracted representation 
\begin{equation*}
(v_0, \min, v_2) (v_2, 7, v_3) (v_3, 4, v_4) (v_4, \min, v_5) (v_5, \max, v_7) (v_7, 5, v_8)\,.
\end{equation*}
For every node $v_0$, there is a one-to-one correspondence between $v_0$-paths and the min-max-contracted representations of $v_0$-paths. 

We extend $\edgeWeight$ to triples $(u, \min, v)$ and $(u, \max, v)$, i.e., $\edgeWeight(u, \min, v) = \edgeWeight(p)$, where $p$ is the unique min-path from $u$ to $v$, and $\edgeWeight(u, \max, v) = \edgeWeight(p)$, where $p$ is the unique max-path from $u$ to $v$. 
If $\tilde{p}$ is the  min-max-contracted representation of the path $p$ then we define $\edgeWeight(\tilde{p}) = \edgeWeight(p)$.
Moreover, we call the values $\edgeWeight(e)$ for all triples $e$ in $\tilde{p}$ the $\edgeWeight$-values of $\tilde{p}$. Hence,
$\edgeWeight(\tilde{p})$ is the sum of all the $\edgeWeight$-values of $\tilde{p}$.

\subparagraph{Data structures.} We describe next certain data structures for $G$ that need to be computed in order to enumerate 
all initial-to-leaf paths. We will see that all these data structures can be computed in time $\bigo(|G|)$, which means that we can compute them for all dags $\DAG_i$ in the preprocessing phase of our final enumeration algorithm in time $|D| \cdot f(d,|\phi|)$. 

The \emph{min-path} of a node $v$ is the $v$-to-leaf path obtained by starting in $v$ and only taking min-edges until we reach a leaf, which we denote by $\min(v)$. In particular, $\min(v) = v$ if and only if $v$ is a leaf. The \emph{max-path} of $v$ and the node $\max(v)$ are defined analogously. We can compute the value $\min(v)$ and  the corresponding weight
$\edgeWeight(v, \min, \min(v))$ as follows
($\max(v)$ and $\edgeWeight(v, \max, \max(v))$ are computed analogously):
For every leaf $v$, we set $\min(v) = v$ and $\edgeWeight(v, \min, \min(v)) = 0$ (actually, the triple $\edgeWeight(v, \min,v)$ does not appear in 
a min-max-contracted representation, it is used here only in order to facilitate the computation of the weights).
For every inner node $v$ with min-edge $(v, 1, u)$, we set $\min(v) = \min(u)$
and $\edgeWeight(v,\min, \min(v)) = \edgeWeight(v,1,u) + \edgeWeight(u,\min, \min(u))$.

For a node $u$ on $v$'s $\min$-path with $v \neq u$, we denote by $\parent_{\min}(v, u)$ the unique parent node of $u$ on $v$'s $\min$-path. Observe that if $u$ is the direct successor of $v$ on $v$'s $\min$-path, then $\parent_{\min}(v, u) = v$. For a node $u$ on $v$'s $\max$-path with $v \neq u$, we define $\parent_{\max}(v, u)$ analogously.

We will next discuss how to compute in time $\bigo(|G|)$ a data structure that allows us to retrieve in constant time the value $\parent_{\min}(v, u)$ for a node $u$ on $v$'s min-path and the value $\parent_{\max}(v, u)$ for a node $u$ on $v$'s max-path. 

Let us only discuss the min-case, since the max-case can be dealt with analogously. We first remove from $G$ all edges that are not min-edges and then we reverse all edge directions. This gives us a forest $F_{\min}$. For every node $v$ of $G$, there is a unique rooted tree $T(v)$ in $F_{\min}$ that contains $v$, and $v$ is the root of $T(v)$ if and only if $v$ is a leaf in $G$. Moreover, the path from $v$ to the root of $T(v)$ is exactly the min-path of $v$ in $G$. This forest $F_{\min}$ can be computed in time $\bigo(|G|)$.

For an arbitrary rooted tree $T$ and nodes $u, v \in T$ such that there is a non-empty path from $u$ to $v$, we denote by $\nextlink_T(u, v)$ the unique child of $u$ that belongs to the unique $u$-to-$v$ path in $T$. A \emph{next link data structure} for $T$ is a data structure that allows us to retrieve $\nextlink_T(u, v)$ for given $u, v \in T$ as above. We call these queries \emph{next link queries}. 
The following result from \cite{GasieniecKPS05} will be used:

\begin{theorem} \label{thm next link}
For a rooted tree $T$, one can construct in time $\bigo(|T|)$ a next link data structure for $T$ that allows to solve next link queries in constant time.
\end{theorem}
For every $v \in V$ and $u \neq v$ on $v$'s min-path we have that $T(v) = T(u)$, and the node $\parent_{\min}(v, u)$ is the unique child of $u$ on the $u$-to-$v$ path in $T(v)$, i.e., $\parent_{\min}(v, u) = \nextlink_{T(v)}(u, v)$. Consequently, by computing the next link data structure
from Theorem~\ref{thm next link} for every tree of the forest $F_{\min}$, we can retrieve $\parent_{\min}(v, u)$ for given $v \in V$ and $u \neq v$ on $v$'s min-path in constant time.

\subparagraph{Enumeration procedure.} 
Provided we have computed the data structures mentioned above, we can support the following operation for $G$. Given the min-max-contracted representation of an initial-to-leaf path $p$, all its $\edgeWeight$-values and its total weight, we can construct in constant time the min-max-contracted representation of the lexicographical successor of $p$ (among all initial-to-leaf paths in $G$)
 along with all its $\edgeWeight$-values and its total weight. If this lexicographical successor does not exist then the algorithm will report 
that $p$ is the lexicographically largest path.

Let $p$ be the min-max-contracted representation of an initial-to-leaf path. We first show how to obtain the min-max-contracted representation of the 
lexicographically next initial-to-leaf path in constant time. To this end, we consider three different cases depending on whether $p$'s last triple is a min-triple, a middle-edge or a max-triple. After that, we discuss how to update the $\edgeWeight$-values and the total weight.
Note that in each of the following cases, $p$ is only modified in a suffix of constant length. Hence, constant time suffices in each case.
In order to avoid some further case distinctions we interpret a triple $(v, \min, v)$ by the empty sequence $\varepsilon$.
Recall that $\min(v) = v$ if and only if $v$ is a leaf.

\medskip

\noindent \textit{Case 1:} $p$ ends with a min-triple. Let $p = p' (v_1, \min, v_2)$.
This means that $v_3 := \parent_{\min}(v_1, v_2)$ is defined and there is an edge $(v_3, 2, v_4)$.

\medskip
\noindent \textit{Case 1.1:} $v_1 = v_3$.
 If $(v_1, 2, v_4)$ is a middle-edge (i.e., $|\gamma(v_1)| > 2$) 
 then we set $p := p' (v_1, 2, v_4) (v_4, \min, \min(v_4))$.
 
 If on the other hand $(v_1, 2, v_4)$ is a $\max$-edge then we have to check whether the last triple of $p'$ (if it exists) is a 
 $\max$-triple or not. If not, then we set $p :=  p' (v_1, \max, v_4) (v_4, \min, \min(v_4))$. 
 On the other hand, if $p'$ is of the form $p' = p'' (v_5, \max, v_1)$ (note that $p'$ ends with $v_1$)
 then we set $p :=  p'' (v_5, \max, v_4)(v_4, \min, \min(v_4))$.
 
 \medskip
\noindent \textit{Case 1.2:} $v_1 \neq v_3$.
 If $(v_3, 2, v_4)$ is a middle-edge (i.e., $|\gamma(v_3)| > 2$) 
 then we set $p := p' (v_1, \min, v_3) (v_3, 2, v_4) (v_4, \min, \min(v_4))$. 
If $(v_3, 2, v_4)$ is a $\max$-edge, then we set 
$p := p' (v_1, \min, v_3) (v_3, \max, v_4)(v_4, \min, \min(v_4))$.

\medskip

\noindent \textit{Case 2:} $p$ ends with middle-edge. Let $p = p' (v_1, j, v_2)$.
Then there exists a middle- or $\max$-edge $(v_1, j + 1, v_3)$ in $G$.
If $(v_1, j + 1, v_3)$ is a middle edge in $G$, then we set $p := p' (v_1, j + 1, v_3) (v_3, \min, \min(v_3))$.

Now assume that $(v_1, j + 1, v_3)$ is a $\max$-edge. If $p'$ does not end with a $\max$-triple
(this includes the case that $p'$ is empty) then we set $p := p'(v_1, \max, v_3) (v_3, \min, \min(v_3))$.
Finally, if $p' = p'' (v_4, \max, v_1)$ ends with a $\max$-triple then we set
$p := p'' (v_4, \max, v_3)(v_3, \min, \min(v_3))$. 

\medskip

\noindent \textit{Case 3:} $p$ ends with a max-triple. If $p = (\iota, \max, v_1)$, then the algorithm reports 
that $p$ is the lexicographically largest path.
Assume now that $p = p' (v_1, j, v_2) (v_2, \max, v_3)$, where $(v_1, j, v_2)$ is a middle-edge or a min-triple 
(in a min-max-contracted representation there do not exist two consecutive max-triples).
We then set $p := p' (v_1, j, v_2)$ and continue with either Case 1 (if $(v_1, j, v_2)$ is a min-triple) or 
Case 2 (if $(v_1, j, v_2)$ is a middle-edge).

\medskip
\noindent
For obtaining the $\edgeWeight$-values of the new min-max-contracted representation, we only have to compute the 
$\edgeWeight$-values  for the new triples (for the old triples the $\edgeWeight$-values remain the same). If a new triple $(u, j, v)$ is a middle edge, then the value $\edgeWeight(u, j, v)$ has been explicitly computed in the preprocessing. 
For triples of the form $(u, \min, v)$, we can argue as follows:
Whenever we add a new min-triple $(u, \min, v)$ in the above algorithm, then one of the following three cases holds:
\begin{itemize}
\item The min-edge $(u,1, v)$ exists in $G$: We then set $\edgeWeight(u, \min, v) = \edgeWeight(u, 1, v)$, where the latter has been precomputed.
\item $v = \min(u)$: Then $\edgeWeight(u, \min, \min(u))$ has been precomputed.
\item There is an old min-triple $(u, \min, v')$ such that $v = \parent_{\min}(u, v')$. In this case we set
$\edgeWeight(u, \min, v) = \edgeWeight(u, \min, v') - \edgeWeight(v,1,v')$. Here $(v,1,v')$ is an edge of $G$ for which the weight is known.
\end{itemize}
For new $\max$-triples we can compute the weights in a similar way.

Now with all the $\edgeWeight$-values of the new $p$ at our disposal, we can easily compute the total weight $\edgeWeight(p)$
for the new $p$ in constant time. Indeed, notice that in the above algorithm, $p$ will be modified 
 by removing a constant number of triples from the end and then adding a constant number of triples at the end.
 Hence, $\edgeWeight(p)$ can be updated by a constant number of subtractions and additions.

 When we speak about the min-max-contracted representation of a path in the following, we assume that 
 all $\edgeWeight$-values as well as  the total weight of the
 path are also computed.
 
 \subparagraph{Enumerating the initial-to-leaf paths of \boldmath{$\DAG_i$}.} 
In order to enumerate the initial-to-leaf paths of $\DAG_i$, we start with the min-max-contracted representation of the first initial-to-leaf path of $\DAG_i$, which can be obtained in constant time, since it is simply $(S, \min, \min(S))$. Then, by the procedure described above, we can obtain the min-max-contracted representation of the lexicographically next initial-to-leaf path from the min-max-contracted representation of the previous initial-to-leaf path in constant time, which gives us a constant delay enumeration algorithm for  the min-max-contracted representations of all such paths. Thereby, also the $\edgeWeight$-values and the total weight are correctly updated.

\subsubsection{Counting the number of \boldmath{$\mathcal{B}_i$}-nodes}
\label{appendix-counting}

For every $\bigR$-neighborhood $\mathcal{B}_i$, we need the total number of $\mathcal{B}_i$-nodes in $\eval(D)$. In the uncompressed setting, this number is given by the size of the explicitly computed list $L_{i}$ (see the algorithm described in Section~\ref{sec:enumAlgo}). In the compressed setting, it can be computed (thanks to Lemma~\ref{bijectionLemma})
as
\[\beta_i = \sum_{A \in N_i} P_A \cdot |\mathcal{L}_{i,A}|,\]
 where $N_i$ is the set of all nonterminals that are $\mathcal{B}_i$-useful, $P_A$ is the number of $S$-to-$A$ paths in $\DAG(D)$, and $\mathcal{L}_{i, A}$ is the list of all valid $\mathcal{B}_i$-nodes in $\expan(A)$. Recall that in Section~\ref{sec:computeneighborhoodTypes}, we have already computed the sets $N_i$ and the lists $\mathcal{L}_{i, A}$.

The numbers $P_A$ can be computed as follows. We first remove from $\DAG(D)$ all nonterminals $A$ such that there is no $S$-to-$A$ path in $\DAG(D)$ (such nonterminals could have been removed in the very beginning without changing $\eval(D)$). This can be done in time $\bigo(|\DAG(D)|) \leq \bigo(|D|)$. 
Then we can compute all numbers $P_A$ inductively top-down in $\DAG(D)$ by first setting $P_S = 1$ and then setting $P_A = \sum_{1 \leq i \leq \ell} P_{A_i}$ for every node $A$ with incoming edges $(A_1, i_1, A), \ldots, (A_\ell, i_\ell, A)$.

The total number of (binary) additions is bounded by the number of edges of $\DAG(D)$, which is bounded by $|D|$.
Moreover, all numbers $P_A$ need only $\bigo(|D|)$ many bits (see Section~\ref{sec node reps}). Therefore, on our machine
model, every addition needs constant time (see Section~\ref{sec:SLPRAM}).
Hence, the whole procedure needs time $\bigo(|\DAG(D)|)$. Therefore, all numbers $\beta_i$ can be computed in time $f(d,|\phi|) \cdot |D|$.

\subsubsection{Checking distance constraints}
\label{sec:disjointnessCheck}

Recall that we fixed the consistent factorization $(\mathcal{B}_1, \sigma_1,  \ldots, \mathcal{B}_m, \sigma_m)$ (see \eqref{eq-consistent-factorization})  of the fixed $(k,r)$-neighbor\-hood type $\mathcal{B}$ at the beginning of Section~\ref{sec:generalStructure}. We also noticed
that during the enumeration phase we have to check whether 
$\dist_{\eval(D)}(t_{b_{i}, \sigma_{i}}, t_{b_{j}, \sigma_{j}}) > 2r+1$ for a $\mathcal{B}_i$-node $b_i$ and a $\mathcal{B}_j$-node $b_j$. Here, the following assumptions
hold for $i$ (and analogously for $j$):
\begin{itemize}
\item $b_i$ is given by a triple $(\lex(p_i), q_i, v_i)$, 
\item $p_i$ is an initial-to-$A_i$ path in $\DAG(D)$ (for some $\mathcal{B}_i$-useful nonterminal $A_i$) that is given by the min-max-contracted representation 
of the corresponding path in $\DAG_i$.
\item $c_i := (q_i,v_i)$ is a node (written in $A_i$-representation) from $\expan(A_i)$ such that $c_i$ has the $\bigR$-neighborhood type $\mathcal{B}_i$ in $\expan(A_i)$ and $\neighborhoodStr{\expan(A_i)}{\bigR}{c_i}$ is a valid substructure of $\expan(A_i)$.
\end{itemize}
Moreover, let $\Gamma_i = \{ (A_i, v) : v \in \ran(\tau_{A_i}) \}$ be the set of contact nodes of $\eval(A_i)$.
Since $(\mathcal{B}_1, \sigma_1, \ldots, \mathcal{B}_m, \sigma_m)$ is a consistent factorization of the $(k,r)$-neighborhood type
$\mathcal{B}$ and $\neighborhoodStr{\expan(A_i)}{\bigR}{c_i} \simeq \mathcal{B}_i$
we have for every $c \in \ran(t_{c_i, \sigma_i})$:
$\sphereStr{\expan(A_i)}{r}{c} \subseteq \sphereStr{\expan(A_i)}{r}{t_{c_i, \sigma_i}} \subseteq  \sphereStr{\expan(A_i)}{\bigR}{c_i}$.
Since $\neighborhoodStr{\expan(A_i)}{\bigR}{c_i}$ is a valid substructure of $\expan(A_i)$, we obtain
$\sphereStr{\eval(A_i)}{r}{c} \cap \Gamma_i = \emptyset$, i.e., 
\begin{equation} \label{dist from contact}
\dist_{\eval(A_i)}(c, \Gamma_i) \geq r+1
\end{equation}
for all $i \in [m]$ and all $c \in \ran(t_{c_i, \sigma_i})$. 
We next  state some sufficient conditions for $\dist_{\eval(D)}(t_{b_{i}, \sigma_{i}}, t_{b_{j}, \sigma_{j}}) > 2r+1$. 

\begin{lemma} \label{SuffCondLemmaOne}
If neither $p_i$ is a prefix of $p_j$ nor the other way around, then we have $\dist_{\eval(D)}(t_{b_{i}, \sigma_{i}}, t_{b_{j}, \sigma_{j}}) > 2r+1$. 
\end{lemma}

\begin{proof}
We have to show that if neither $p_i$ is a prefix of $p_j$ nor the other way around, then $\dist_{\eval(D)}(b, b') > 2r+1$ for every $b \in \ran(t_{b_{i}, \sigma_{i}})$ and $b' \in \ran(t_{b_{j}, \sigma_{j}})$. To this end, let $b \in \ran(t_{b_{i}, \sigma_{i}})$ and $b' \in \ran(t_{b_{j}, \sigma_{j}})$ be arbitrarily chosen. If $b$ and $b'$ are not connected in $\mathcal{G}(\eval(D))$, then $\dist_{\eval(D)}(b, b') = \infty > 2r+1$. So, let us assume that $b$ and $b'$ are connected. Let $c \in \ran(t_{c_{i}, \sigma_{i}})$ and $c' \in \ran(t_{c_{j}, \sigma_{j}})$  
such that $\eta_{p_i}(c) = b$ and $\eta_{p_j}(c') = b'$.

If neither $p_i$ is a prefix of $p_j$ nor the other way around, then the intersection of the (universes of the) substructures $\eta_{p_i}(\eval(A_i))$ and $\eta_{p_j}(\eval(A_j))$ is $\eta_{p_i}(\Gamma_i) \cap \eta_{p_j}(\Gamma_j)$, which can be non-empty.
Consequently, any path $\Pi$ from $b$ to $b'$ has a shortest prefix $\Pi_1$ that goes from $b = \eta_{p_i}(c)$ to $\eta_{p_i}(\Gamma_i)$ as well as
a shortest suffix $\Pi_2$ that goes from 
 $\eta_{p_j}(\Gamma_j)$ to  $b' = \eta_{p_j}(c')$. The path $\Pi_1$ must be the  $\eta_{p_i}$-image of a path from 
 $c$ to $\Gamma_i$ in $\eval(A_i)$ and the path $\Pi_2$ must be the $\eta_{p_j}$-image of a path from $\Gamma_j$ to $c'$ in $\eval(A_j)$.
From \eqref{dist from contact} it follows that $\Pi_1$ and $\Pi_2$ have both length at least $r+1$. Hence,
every path from $b$ to $b'$ has length at least $2(r+1) = 2r+2$, i.e., $\dist_{\eval(D)}(b, b') > 2r+1$. 
This shows the lemma.
\end{proof}

\begin{lemma} \label{lemma-dist-implication}
Assume that $p_j = p_i q$ for some path $q$.
Let $b \in \ran(t_{b_{i}, \sigma_{i}})$ and $b' \in \ran(t_{b_{j}, \sigma_{j}})$
and take the unique 
 $c \in \ran(t_{c_{i}, \sigma_{i}})$ and $c' \in \ran(t_{c_{j}, \sigma_{j}})$ such that 
$\eta_{p_i}(c) = b$ and $\eta_{p_j}(c') = \eta_{p_i}(\eta_q(c')) =  b'$.
If $\dist_{\eval(D)}(b,b') \leq 2r+1$ then $\dist_{\eval(A_i)}(c,\eta_q(c')) \leq 2r+1$. 
\end{lemma}

\begin{proof}
Since $p_j = p_i q$, we know that $\eta_q$ is an embedding of $\eval(A_j)$ into $\eval(A_i)$.

There is a path $\Pi$ between $b$ and $b'$ of length at most $2r+1$ in the graph $\mathcal{G}(\eval(D))$. 
We first show that $\Pi$ does not contain a node from $\eta_{p_i}(\Gamma_i)$.
If this would not be the case then we can take the shortest prefix $\Pi_1$ of $\Pi$ that goes from $b = \eta_{p_i}(c)$ to
$\eta_{p_i}(\Gamma_i)$ and the shortest suffix $\Pi_2$ of $\Pi$ that goes from $\eta_{p_j}(\Gamma_j) = \eta_{p_i}(\eta_q(\Gamma_j))$ to 
$b' = \eta_{p_j}(c') = \eta_{p_i}(\eta_q(c'))$. 
The path $\Pi_1$ is contained in $\eta_{p_i}(\eval(A_i))$ and visits a node from $\eta_{p_i}(\Gamma_i)$ only at the very end.
Similarly, the path $\Pi_2$ is contained in $\eta_{p_j}(\eval(A_j))$ and visits a node from  $\eta_{p_j}(\Gamma_j)$
only in the beginning.
By \eqref{dist from contact} we have 
 $\dist_{\eval(A_i)}(c, \Gamma_i) \geq r+1$ and $\dist_{\eval(A_j)}(c', \Gamma_j) \geq r+1$, which implies
 that $\Pi_1$ and $\Pi_2$ have both lenght at least $r+1$. 
Hence, $\Pi$  has length at least $2r+2$, which is a contradiction.

 Therefore, $\Pi$ does not visit a node from $\eta_{p_i}(\Gamma_i)$, which implies that 
$\Pi$ is a path in $\eta_{p_i}(\eval(A_i))$ and does not contain edges between nodes from 
$\eta_{p_i}(\Gamma_i)$. Hence, $\Pi$ arises from a path between $c$ and $\eta_q(c')$ in $\eval(A_i)$,
which finally yields $\dist_{\eval(A_i)}(c, \eta_{q}(c')) \leq 2r+1$. 
\end{proof}

\begin{lemma}\label{SuffCondLemmaTwo}
If $p_j = p_i q$ or $p_i = p_j q$, 
and the path $q$ has length at least $3\bigR - r +1$, then $\dist_{\eval(D)}(t_{b_{i}, \sigma_{i}}, t_{b_{j}, \sigma_{j}}) > 2r+1$.
\end{lemma}

\begin{proof}
Assume that $p_j = p_i q$ for some path $q$ of length at least $3\bigR - r +1$. In addition, assume that
$\dist_{\eval(D)}(t_{b_{i}, \sigma_{i}}, t_{b_{j}, \sigma_{j}}) \leq 2r+1$, which means that $\dist_{\eval(D)}(b, b') \leq 2r+1$ for some $b \in \ran(t_{b_{i}, \sigma_{i}})$ and $b' \in \ran(t_{b_{j}, \sigma_{j}})$. We will deduce a contradiction. 
Let $c \in \ran(t_{c_{i}, \sigma_{i}})$ and $c' \in \ran(t_{c_{j}, \sigma_{j}})$ such that $\eta_{p_i}(c) = b$ and $\eta_{p_j}(c') = b'$. Since $p_j = p_i q$, we know that $\eta_q$ is an embedding of $\eval(A_j)$ into $\eval(A_i)$.

By Lemma~\ref{lemma-dist-implication} we have $\dist_{\eval(A_i)}(c, \eta_{q}(c')) \leq 2r+1$. 
Since the neighborhoods  $\neighborhoodStr{\expan(A_i)}{r}{t_{c_i, \sigma_i}}$ and $\neighborhoodStr{\expan(A_j)}{r}{t_{c_j, \sigma_j}}$ are connected,
it follows that $\neighborhoodStr{\expan(A_i)}{r}{t_{c_i, \sigma_i} \sqcup  \eta_q \circ t_{c_j, \sigma_j}} \subseteq \eval(A_i)$ is connected.
Since $t_{c_i, \sigma_i} \sqcup  (\eta_q \circ t_{c_j, \sigma_j})$ is a partial $k$-tuple,
we obtain
\[ 
\dist_{\eval(A_i)}(c_i, \eta_{q}(c_j)) \leq (2r+1) (k-1) = (2rk - r + k - 1) - r = \bigR - r.
\]
Since $\neighborhoodStr{\expan(A_i)}{\bigR}{c_i} \simeq \mathcal{B}_i$ and $\neighborhoodStr{\expan(A_j)}{\bigR}{c_j} \simeq \mathcal{B}_j$
 are valid substructures of $\expan(A_i)$ and $\expan(A_j)$, respectively, we know that $\dist_{\eval(A_i)}(\mathsf{In}_{A_i}, c_i) \leq \bigR$ and $\dist_{\eval(A_j)}(\mathsf{In}_{A_j}, c_j) \leq \bigR$. 
 Consequently, $\dist_{\eval(A_i)}(\mathsf{In}_{A_i}, \eta_{q}(\mathsf{In}_{A_j})) \leq 2\bigR + \bigR - r = 3\bigR - r$.
This is a contradiction, since the fact that $q$ has length at least $3\bigR - r +1$ means that $\dist_{\eval(A_i)}(\mathsf{In}_{A_i}, \eta_{q}(\mathsf{In}_{A_j})) > 3\bigR - r$ due to the apex
condition for $D$.

We conclude that $\dist_{\eval(D)}(b, b') > 2r+1$ for every $b \in \ran(t_{b_{i}, \sigma_{i}})$ and $b' \in \ran(t_{b_{j}, \sigma_{j}})$, which means that $\dist_{\eval(D)}(t_{b_{i}, \sigma_{i}}, t_{b_{j}, \sigma_{j}}) > 2r+1$.
\end{proof}
By Lemmas~\ref{SuffCondLemmaOne} and \ref{SuffCondLemmaTwo} it remains to show the following:
\begin{enumerate}[(i)]
\item \label{point-i}  We can check in time $f(d,|\phi|)$ whether $p_j = p_i q$ or $p_i = p_j q$ for some
path $q$ of length at most $3 \bigR - r$.
\item \label{point-ii} Assuming (i) holds we can check in time $f(d,|\phi|)$ whether $\dist_{\eval(D)}(t_{b_{i}, \sigma_{i}}, t_{b_{j}, \sigma_{j}}) > 2r+1$. 
\end{enumerate}
Let us start with \eqref{point-ii} and assume that $p_j = p_i q$ (the case $p_i = p_j q$ can be handled analogously)
with $q$ of length at most $3 \bigR - r$. 
In order to check whether $\dist_{\eval(D)}(t_{b_{i}, \sigma_{i}}, t_{b_{j}, \sigma_{j}}) > 2r+1$, we have to check whether $\dist_{\eval(D)}(b, b') > 2r+1$ for all $b \in \ran(t_{b_{i}, \sigma_{i}})$ and $b' \in \ran(t_{b_{j}, \sigma_{j}})$. We explain how this can be done for fixed $b \in \ran(t_{b_{i}, \sigma_{i}})$ and $b' \in \ran(t_{b_{j}, \sigma_{j}})$. 
As before, let $c \in \ran(t_{c_{i}, \sigma_{i}})$ and $c' \in \ran(t_{c_{j}, \sigma_{j}})$ such that 
$\eta_{p_i}(c) = b$ and $\eta_{p_j}(c') = \eta_{p_i}(\eta_q(c')) =  b'$.
By Lemma~\ref{lemma-dist-implication},
$\dist_{\eval(D)}(b,b') \leq 2r+1$ implies $\dist_{\eval(A_i)}(c,\eta_q(c')) \leq 2r+1$. Since the other direction is trivial,
it suffices for \eqref{point-ii} to check $\dist_{\eval(A_i)}(c,\eta_q(c')) > 2r+1$.
  Note that the $A_i$-representation of $\eta_{q}(c')$ can be computed in time $\bigo(\bigR)$: $c'$ is given in its
$A_j$-representation $(q', v')$, hence, the  $A_i$-representation of $\eta_{q}(c')$ is $(qq',v')$, where $q$ has length at most $3 \bigR - r$.

In order to check $\dist_{\eval(A_i)}(c, \eta_{q}(c')) > 2r+1$, we construct $\neighborhoodStr{\eval(A_i)}{2r+1}{c}$ by starting a BFS in $c$ and then computing all elements of $\eval(A_i)$ with distance at most $2r + 1$ from $c$. This is analogous to the construction of the expansions in Section~\ref{sec compute expansion}. Then, we check whether $\eta_{q}(c') \in \neighborhoodStr{\eval(A_i)}{2r+1}{c}$. Since both $c$ and $\eta_{q}(c')$ are given by $A_i$-representations of size  $f(d,|\phi|)$, this can be done in time $f(d,|\phi|)$.

For \eqref{point-i} we show that one can check in time $f(d, |\phi|)$ whether $p_i = p_j q$ for a path $q$ of length at most $3 \bigR - r$. The general idea is to check whether $p_i = p_j$ and, if this is not the case, to remove repeatedly the last edge of $p_i$ (for at most $3 \bigR - r$ times) and check whether the resulting path equals $p_j$. Recall that the path $p_i$ is stored as the min-max-contracted representation $\tilde{p}_i$ (in $\DAG_i$) of the corresponding path $p_i 0A'_i$ in $\DAG(D)'$, where in addition all 
maximal non-branching paths 
in $p_i 0A'_i$ haven been contracted, and analogously for $p_j$ (see Section~\ref{sec:pathEnumPreprocessing}).
Checking $p_i = p_j$ can be done by simply checking whether $\lex(p_i) = \edgeWeight(\tilde{p}_i) = \edgeWeight(\tilde{p}_j) = \lex(p_j)$; see \eqref{eq-lex-weight}. Hence, it remains to show how we can obtain from $\tilde{p}_i$ a min-max-contracted representation $\textsf{shorten}(\tilde{p}_i)$, which
represents the path obtained from $p_i$ by removing the last edge (using this, one can also remove in a first step  the terminal edge $(A_i, 0, A'_i)$ that was added to $\DAG(D)$ in the construction of $\DAG(D)'$).
Here we have to consider two aspects:
\begin{itemize}
\item Maximal non-branching subpaths have been contracted into single edges in the construction of $\DAG_i$ from $\DAG(D)'$.
\item Maximal subpaths consisting of $\min$-edges ($\max$-edges, respectively) have been contracted.
\end{itemize}
To address both aspects, we compute $\textsf{shorten}(\tilde{p}_i)$ in two steps. Let $\tilde{p}_i = \tilde{q} (B, \ell, A)$ 
(we assume that we have already done some shortening steps, so that the $A$ might be no longer $A'_i$).

\medskip
\noindent
\emph{Step 1:}
 If $(u, \ell, v)$ is neither a $\min$- nor a $\max$-triple, then $(B, \ell, A)$ is an edge in $\DAG_i$ and we pass
$\tilde{p}' := \tilde{q}$ together with the removed edge $(B, \ell, A)$ to Step~2 below.
If $\ell = \min$ and  $(B, 1, A)$ is an edge of $\DAG_i$ then we again pass
$\tilde{p}' := \tilde{q}$ and the edge $(B,1,A)$ to Step~2. If $(B, 1, A)$ is not an edge of $\DAG_i$, then 
$(\parent_{\min}(B, A), 1, A)$ is  an edge in $\DAG_i$ and we pass $\tilde{p}' := \tilde{q} (B, \min, \parent_{\min}(B, A))$ and the edge $(\parent_{\min}(B, A), 1, A)$ to Step 2. Finally, the case $\ell = \max$ can be handled  analogously. 
 Just like described in the enumeration procedure of Section~\ref{sec:MainPathEnumSection}, we can also compute the $\edgeWeight$-values and total weight for $\tilde{p}'$. 

\medskip
\noindent
\emph{Step 2:} From Step~1 we obtain a min-max-contracted path representation $\tilde{p}'$ together with an edge $(B', \ell, A')$ from
$\DAG_i$. In general we removed from $\tilde{p}$ not only a single edge in $\DAG(D)'$ but a maximal non-branching subpath
that is represented by the $\DAG_i$-edge $(B', \ell, A')$.
Therefore, we have to add to $\tilde{p}'$ the prefix of this maximal non-branching subpath without the last edge. 
To do this,  we store in the preprocessing for every edge in $\DAG_i$ the information whether it is a single edge of $\DAG(D)'$ or an edge that represents a contracted maximal non-branching path of length at least 2. 
If $(B', \ell, A')$ is a single edge from $\DAG_i$, then we are done and set $\textsf{shorten}(\tilde{p}_i) = \tilde{p}'$.
If $(B', \ell, A')$ is an edge of $\DAG_i$ that represents a maximal non-branching path from $B'$ to $A'$ of length at least $2$ in $\DAG(D)'$, then we set $\textsf{shorten}(\tilde{p}_i) = \tilde{p}'(B', \ell, A'')$, where $A''$ is the parent node of $A'$ on the maximal non-branching path from $B'$ to $A'$ (and we also store the information whether $(B', \ell, A'')$ represents still a maximal non-branching path of length at least $2$).

This correctly constructs the min-max-contracted representation $\textsf{shorten}(\tilde{p}_i)$ for the path $p_i$ with the last edge removed. However, it remains to explain how the construction from Step 2 can be implemented in constant time. This can again be achieved by a next link data structure; see Theorem~\ref{thm next link}. In the preprocessing, we compute a forest of all the reversed maximal non-branching paths of $\DAG(D)'$ and then we compute a next link data structures for all the trees in this forest. 
So, in some sense we use a two-level next-link data structure: the upper level handles contracted sequences of min-edges and max-edges, respectively, whereas the lower level handles contracted sequences of maximal non-branching edges.

\begin{remark}
Note that the above shortening of $p_i$ destroys the representation of the path $p_i$ (and similarly for $p_j$).
This is a problem, since later in the enumeration phase $p_i$ might be needed again for comparison with other paths (and the number of these comparisons depends on the structure $\eval(D)$).
Producing a copy of $p_i$ before we start shortening $p_i$ is not an option since the min-max-contracted representation of $p_i$ does not fit into a constant
number of registers. Therefore making a copy of $p_i$ would flaw the constant delay requirement.
Fortunately, there is a simple solution. Whenever we remove a terminal edge from the min-max-contracted representation
of the path $p_i$ we store this edge. In total we have to store only $3 \rho - r$ many
edges. Then, when we want to restore $p_i$ we add the removed edges at the end, which can be easily done with the min-max-contracted
representations (it is similar to the shortening procedure).
\end{remark}

\begin{remark} \label{remark disjoint spheres}
The algorithm from this section also yields the last piece for checking in the preprocessing phase, whether a sentence of 
the form \eqref{eq-local-sentence} holds (see the last paragraph in Appendix~\ref{sec:GaifmanReductionCompressed}).
Recall that \eqref{eq-local-sentence} expresses that there exists at least $q = f(|\phi|)$ many nodes $a_1, \ldots, a_q$ with the following properties: $\dist_{\eval(D)}(a_i, a_j) > 2r$ and the $r$-neighborhood type
of each $a_i$ is from a fixed set of $r$-neighborhood types.
\end{remark}

\subsubsection{Final Remarks}
\label{sec:final-remarks}
For a fixed consistent factorization $\Lambda = (\mathcal{B}_1, \sigma_1, \ldots, \mathcal{B}_m, \sigma_m)$ of a $(k,r)$-neighborhood type $\mathcal{B}$,
the algorithm keeps a min-max-contracted initial-to-leaf path in $\DAG_i$ for every $i \in [m]$. We have seen in Section~\ref{sec:MainPathEnumSection} how all initial-to-leaf paths of $\DAG_i$ can be enumerated by exploiting the min-max-contracted representation. In Appendix~\ref{sec:disjointnessCheck}, we discussed how the stored min-max-contracted initial-to-leaf paths can be used in order to check the disjointedness of the $r$-neighborhoods of the produced tuples. Consequently, we can perform Algorithm~\ref{mainEnumAlgoNonCompressed} also in the compressed setting.

Recall from Theorem~\ref{prelimMainResultTheorem-2} that nodes of $\eval(D)$ are output in lex-presentation. 
Internally, the algorithm represents a node $b \in \eval(D)$ by a triple $(\lex(p),q,v)$, where the initial path $p$
is given by a min-max-contracted representation, and $(q,v)$ is the $A$-representation ($A$ is the nonterminal, where $p$ ends)
of a node from the expansion $\expan(A)$. Hence, $q$ is a path in $\DAG(D)$ of length at most $2 \bigR+1$ (due to the apex condition)
We can therefore compute from the precomputed edge weights also $\edgeWeight(q)$ with $2 \bigR$ binary additions.
Since $(pq,v)$ is the $D$-representation of the node $b$, its lex-representation is $(\lex(pq), v) = (\lex(p)+\lex_A(q), v) = 
(\edgeWeight(p)+\edgeWeight(q)+1, v)$
which can be computed in time $\bigo(\bigR)$ from the internal representation $(\lex(p),q,v)$ of $b$.

\end{document}